\documentclass[submission,copyright,creativecommons]{eptcs}
\providecommand{\event}{AiML 2026} 

\usepackage{aiml26}

\usepackage{iftex}

\ifpdf
  \usepackage{underscore}         % Only needed if you use pdflatex.
  \usepackage[T1]{fontenc}        % Recommended with pdflatex
\else
  \usepackage{breakurl}           % Not needed if you use pdflatex only.
\fi

\title{Intuitionistic Monotone Modal Logic:\\ Proof Theory and Semantics}
\author{Tiziano Dalmonte\footnote{Tiziano Dalmonte acknowledges the financial support through the ‘Abstractron’ project funded by the Autonome Provinz Bozen - Südtirol (Autonomous Province of Bolzano/Bozen) through the Research Südtirol/Alto Adige 2022 Call.}
\institute{Free University of Bozen-Bolzano \\ Bozano, Italy}
\email{tiziano.dalmonte@unibz.it}
\and
Jim de Groot\footnote{Jim de Groot was supported by Swiss National Science Foundation (SNSF) grant No.~200021\_215157.}
\institute{University of Bern\\
Bern, Switzerland}
\email{jim.degroot@unibe.ch}
}

\newcommand{\titlerunning}{Intuitionistic Monotone Modal Logic: Proof Theory and Semantics}
\newcommand{\authorrunning}{T. Dalmonte \& J. de Groot}

\hypersetup{
  bookmarksnumbered,
  pdftitle    = {\titlerunning},
  pdfauthor   = {\authorrunning},
  pdfsubject  = {EPTCS},               % Consider adding a more appropriate subject or description
  pdfkeywords = {Intuitionistic modal logic, monotone modal logic, neighbourhood semantics, sequent calculus} % Uncomment and enter keywords specific to your paper
}

\newcommand{\varax}{\hilbertaxiomstyle{X}}
\newcommand{\varrule}{\mf{x}}

\newcommand{\G}{\Gamma}
\newcommand{\D}{\Delta}
\newcommand{\Dp}{\D}
\newcommand{\Sp}{\Sigma}
\newcommand{\seq}{\Rightarrow}
\newcommand{\sseq}{\Leftrightarrow}
\newcommand{\seqder}[1]{\overset{\raise.3em\hbox{$\varDer_{#1}$}}{\seq}}
\newcommand{\sseqder}[2]{\overset{\raise.3em\hbox{$\varDer_{#1}, \varDer_{#2}$}}{\sseq}}

\newcommand{\vd}{\vdash}

\newcommand{\mc}{\mathcal}
\newcommand{\mf}{\mathsf}

\newcommand{\varDer}{\mathcal{D}}
\newcommand{\neigh}{N}

\newcommand{\imp}{\to}

\newcommand{\diam}{\Diamond}
\newcommand{\lan}{\mathcal{L}}

\newcommand{\fintf}{\iota}
\newcommand{\fint}{i}
\newcommand{\logicnamestyle}[1]{\ensuremath{\mathsf{#1}}}
\newcommand{\K}{\logicnamestyle{K}}

\newcommand{\IPL}{\logicnamestyle{IPL}}
\newcommand{\WK}{\logicnamestyle{WK}}
\newcommand{\CK}{\logicnamestyle{CK}}
\newcommand{\IK}{\logicnamestyle{IK}}
\newcommand{\sfour}{\logicnamestyle{S4}}

\newcommand{\CM}{\logicnamestyle{CM}}
\newcommand{\WM}{\logicnamestyle{WM}}
\newcommand{\IM}{\logicnamestyle{IM}}

\newcommand{\IMstar}{\logicnamestyle{IM} \oplus \Ax}
\newcommand{\IMpK}{\logicnamestyle{IM} \oplus \axK}

\newcommand{\IMN}{\logicnamestyle{IMN}}
\newcommand{\IMP}{\logicnamestyle{IMP}}
\newcommand{\IMD}{\logicnamestyle{IMD}}
\newcommand{\IMT}{\logicnamestyle{IMT}}

\newcommand{\vdIM}{\vd_{\IM}}

\newcommand{\CIM}{\mathcal{C}_{\IM}}
\newcommand{\CIMstar}{\mathcal{C}_{\IMstar}}
\newcommand{\CIMpK}{\mathcal{C}_{\IMpK}}
\newcommand{\CWM}{\mathcal{C}_{\WM}}
\newcommand{\HM}{\mathsf{HM}}
\newcommand{\EM}{\mathsf{M}}

\newcommand{\less}{\futs}

\newcommand{\futs}{\leq}
\newcommand{\R}{R}

\newcommand{\ax}{\AxiomC}
\newcommand{\uinf}{\UnaryInfC}
\newcommand{\binf}{\BinaryInfC}

\newcommand{\llab}{\LeftLabel}
\newcommand{\rlab}{\RightLabel}
\newcommand{\disp}{\DisplayProof}

\newcommand{\bl}[1]{{#1}^\bullet}
\newcommand{\wh}[1]{{#1}^\circ}
\newcommand{\nest}[1]{\{#1\}}
\newcommand{\enest}{\nest{ \ }}
\newcommand{\pr}[1]{{#1}^{\downarrow}}
\newcommand{\thi}[1]{{#1}^{\Downarrow}}
\newcommand{\B}{\mathcal B}
\newcommand{\Der}{\mathcal{D}}

\newcommand{\Rule}{R}
\newcommand{\Rulein}{\Rule_{in}}

\newcommand{\id}{\mf{id}}
\newcommand{\bbot}{\bl{\bot}}
\newcommand{\bimp}{\bl{\imp}}
\newcommand{\wimp}{\wh{\imp}}

\newcommand{\bland}{\bl{\land}}
\newcommand{\wland}{\wh{\land}}
\newcommand{\blor}{\bl{\lor}}
\newcommand{\wlor}{\wh{\lor}}
\newcommand{\bbox}{\bl{\Box}}
\newcommand{\wbox}{\wh{\Box}}
\newcommand{\bdiam}{\bl{\diam}}
\newcommand{\wdiam}{\wh{\diam}}

\newcommand{\cut}{\mf{cut}}
\newcommand{\bwk}{\bl{\mathsf{wk}}}
\newcommand{\wwk}{\wh{\mathsf{wk}}}
\newcommand{\bctr}{\bl{\mathsf{ctr}}}
\newcommand{\nctr}{\mathsf{ctr}^{\nest{\bullet}}}
\newcommand{\bsub}{\bl{\mathsf{sub}}}
\newcommand{\wsub}{\wh{\mathsf{sub}}}
\newcommand{\nelim}{\mathsf{elim}}

\newcommand{\swk}{\mathsf{wk}^*}

\newcommand{\rulen}{\mf{n}}
\newcommand{\rulet}{\mf{t}}
\newcommand{\rulep}{\mf{p}}
\newcommand{\ruled}{\mf{d}}

\newcommand{\rulek}{\mf{k}}

\newcommand{\hilbertaxiomstyle}[1]{\ensuremath{\mathsf{#1}}}

\newcommand{\axN}{\hilbertaxiomstyle{N}}
\newcommand{\axK}{\hilbertaxiomstyle{K}}
\newcommand{\axT}{\hilbertaxiomstyle{T}}
\newcommand{\axD}{\hilbertaxiomstyle{D}}
\newcommand{\axP}{\hilbertaxiomstyle{P}}

\newcommand{\axW}{\hilbertaxiomstyle{W}}
\newcommand{\RD}{\hilbertaxiomstyle{RD}}

\newcommand{\axCbox}{\hilbertaxiomstyle{C}_\Box}

\newcommand{\axPbox}{\hilbertaxiomstyle{P}_\Box}
\newcommand{\axKbox}{\hilbertaxiomstyle{K}_\Box}
\newcommand{\axKbd}{\hilbertaxiomstyle{K}}

\newcommand{\axKdiam}{\hilbertaxiomstyle{K}_\diam}
\newcommand{\axPdiam}{\hilbertaxiomstyle{P}_\diam}
\newcommand{\axIdiam}{\hilbertaxiomstyle{I}_\diam}
\newcommand{\axTbox}{\hilbertaxiomstyle{T}_{\Box}}
\newcommand{\axTdiam}{\hilbertaxiomstyle{T}_{\Diamond}}
\newcommand{\axTbd}{\hilbertaxiomstyle{T}}

\newcommand{\monbox}{\hilbertaxiomstyle{mon}_\Box}
\newcommand{\mondiam}{\hilbertaxiomstyle{mon}_\diam}
\newcommand{\str}{\hilbertaxiomstyle{str}}

\newcommand{\axmnc}{\hilbertaxiomstyle{mnc}}

\usepackage{graphicx}
\usepackage{amsmath}
\usepackage{amssymb}

\usepackage{mathrsfs}
\usepackage{bussproofs}
\usepackage{multirow}
\usepackage{multicol}
\usepackage{mathtools}
\usepackage{tikz}
  \usetikzlibrary{arrows,arrows.meta}

\usepackage{mathrsfs}
\usepackage[cal=txupr,scr=boondoxupr,scrscaled=1.050]{mathalfa}

\usepackage{wasysym}

\DeclareMathOperator{\Ax}{Ax}
\newcommand{\Prop}{\mathrm{Prop}}
\newcommand{\MForm}{\mathcal{L}^{\raisebox{1pt}{\scalebox{.5}{\RIGHTcircle}}}}
\newcommand{\Blocks}{\mathcal{L}^{\{\}}}
\newcommand{\mo}{\mathfrak}
\newcommand{\ms}{\mathscr}

\newcommand{\power}{\mathcal{P}}
\renewcommand{\iff}{\quad\text{iff}\quad}
\renewcommand{\phi}{\varphi}
\newcommand{\Mon}{\triangledown}
\newcommand{\up}{\mathrm{up}}

\newcommand{\subsetsim}{\mathrel{\substack{\textstyle\subset\\[-0.2ex]\textstyle\sim}}}

\newcommand{\PF}{\mathrm{PF}}

\usepackage{fontawesome}

\renewcommand{\star}{\mathbin{\scalebox{.75}{\text{\faStarO}}}}

\usepackage{phaistos}

\begin{document}
\maketitle

\begin{abstract}
We study the intuitionistic monotone modal logic $\IM$ recently introduced in \cite{Gro25-im}.
We first provide a semantic characterisation for a family of natural extensions of $\IM$ in terms of constructive neighbourhood models. 
We then present a calculus for $\IM$ and its extensions, 
obtained by adapting a structured calculus for the classical monotone modal logic $\EM$.
Based on the calculus, we prove some preliminary results for $\IM$, including its decidability. 
Our calculus also reveals an interesting analogy between constructive and intuitionistic variants of $\EM$ and
the corresponding variants of $\K$, thereby 
further justifying $\IM$ as a faithful intuitionistic variant of $\EM$.
\end{abstract}

%%%%%%%%%%%%%%%%%%%%%%%%%%%%%%%%%%%%%%%%%%%%%%%%%%%%%%%%%%%%%%%%%%%%%%%%%%%%%%%%
\section{Introduction}

  Constructive or intuitionistic modal logics arise from extending
  intuitionistic logic with modal operators, such as necessity or
  possibility modalities.
  While early work in the area is motivated by mathematical
  curiosity and the more philosophical
  ``attempt to find modal logics acceptable to intuitionists''~\cite{Bul65b,Bul66},
  intuitionistic modal logics soon proved useful much more broadly.
  For example, Goldblatt developed an intuitionistic modal logic
  to reason about the Grothendieck topology~\cite{Gol81},
  and various constructive systems of epistemic reasoning were
  studied in e.g.~\cite{Wil92,Hir10,Pro12,JagMar16,ArtPro16}.
  Moreover, intuitionistic modal logics are used to model countless
  computational phenomena, with applications ranging from
  hardware verification~\cite{FaiMen97} to
  access control~\cite{GarPfe06} to
  staged computation~\cite{DavPfe96,DavPfe01,NanPfePie08},
  and from type inhabitation of arrows in functional
  programming~\cite{Hug00,LinWadYal11},~\cite[Section~7.1]{LitVis18}
  to program evaluation~\cite{Pit91}
  to knowledge representation in AI~\cite{McC93,McCBuv94,Pai03,MenPai05}.

  A peculiar feature of intuitionistic modal logics is that,
  similar to the intuitionistic connectives, the modal operators
  need not be interdefinable. This disconnection often gives rise to
  a wide spectrum of intuitionistic counterparts of the same classical
  system.
  On one side of this spectrum sit for example logics studied by 
  Sotirov~\cite[Section~4]{Sot84}, Wolter and Zakharyaschev~\cite[Section~2]{WolZak99},
  and Dalmonte, Grellois and Olivetti~\cite{DalGreOli20},
  in which $\Box$ and $\Diamond$ are completely unrelated.
  At the other extremity, we have systems
  studied by Bull~\cite{Bul65b} and
  Bo\v{z}i\'{c} and Do\v{s}en~\cite[Section~11]{BozDos84},
  in which the diamond operator can be viewed as an abbreviation for $\neg\Box\neg$.

  The study of constructive analogues of classical normal modal logics,
  such as $\K$ and $\sfour$, has resulted in two prominent ways to establish
  a more subtle interaction between the modalities.
  First, \emph{constructive} counterparts are obtained by restricting
  a sequent calculus for the classical modal logic in question to (at most)
  one conclusion.
  This originates from Wijesekera~\cite{Wij90}, and gives rise to
  analogues of $\K$ such as $\CK$~\cite{BelPaiRit01}, $\WK$~\cite{Wij90,WijNer05}
  and $\mathsf{CS4}$~\cite{AleEA01}.
  Second, \emph{intuitionistic} counterparts of normal modal logics
  are obtained by taking the set of formulas whose
  standard translation is derivable in \emph{intuitionistic} first-order
  logic~\cite{Fis77,Sim94}.
  This gives rise to \emph{intuitionistic $\K$} ($\IK$) studied by~\cite{Fis77,Fis81,PloSti86},
  as well as intuitionistic counterparts of other normal modal logics
  including $\mathsf{IS4}$ and $\mathsf{IKD}$~\cite{Ewa86,Sim94}.

  Besides intuitionistic analogues of $\K$,
  extensions of intuitionistic logic with monotone modal operators have
  been studied as well. Monotone modal operators are modalities $\Mon$ for
  which $\phi \to \psi$ implies $\Mon\phi \to \Mon\psi$, or equivalently (in presence of the congruence rule),
  that satisfy $\Mon(\phi \wedge \psi) \to \Mon\phi$.
  The study of these in an intuitionistic context dates back to
  Bull~\cite[Page~5]{Bul65}, who studied a single monotone modality that
  satisfies the S4 axioms.
  Discarding the S4 axioms yields Goldblatt's intuitionistic logic
  with a geometric modality~\cite{Gol81,Gol93},
  which was further researched in e.g.~\cite{Gro22gt,GroPat20,TabIemJal22}.
  However, compared to the wealth of research on intuitionistic analogues of $\K$,
  the study of intuitionistic versions of the classical monotone modal logic $\EM$
  is still in its infancy.
  This is surprising, because classical monotone modal logic has been
  researched extensively \cite{Che80,Han03,HanKup04,SanVen10,Fri14}
  and encompasses modal logics such as
  game logic~\cite{Par85},
  concurrent propositional dynamic logic~\cite{Gol87},
  coalition logic~\cite{Pau01,Pau02}
  and alternating-time temporal logic~\cite{AluHenKup02}.

  The study of intuitionistic monotone modal logics with two modalities is
  even more recent.
  The analogy with the normal modal case gives three ways to decide on the level of
  interaction between the modalities.
  First, restricting a (suitable) sequent calculus for $\EM$ to precisely one
  conclusion yields a logic in which the modal operators are completely
  independent~\cite{Dal25}. The resulting logic, called $\CM$ in analogy with $\CK$,
  may be seen as a bimodal version of Goldblatt's intuitionistic logic with a geometric modality.
  Second, the restriction of a sequent calculus for $\EM$ to at most one
  conclusion (i.e.~to zero or one conclusion) is investigated in~\cite{Dal22}, 
  and gives rise to the logic $\WM$ (and extensions).\footnote{The logic $\WM$ from~\cite{Dal22} is called $\IM$ in~\cite{DalGreOli20}. We reserve the name $\IM$ for the logic from~\cite{Gro25-im}. This aligns with the naming convention for constructive and intuitionistic normal modal logics~\cite{GroShiClo25}.}
  Finally, in~\cite{Gro25-im} an \emph{intuitionistic} approach was taken:
  classical monotone modal logic was translated into a suitable first-order
  logic, and changing this first-order logic to an intuitionistic first-order
  logic gave rise to the intuitionistic monotone modal logic $\IM$.
  This mimics the method of obtaining $\IK$ from $\K$ using intuitionistic
  first-order logic~\cite{Sim94}.
  Given its recent introduction, the study of $\IM$ is still
  lacks a proof theoretic treatment, a study of its extensions,
  and basic decidability results.
  This is the gap we aim to close in this paper.
  
  We start by equipping $\IM$ with a semantics by means of
  \emph{constructive neighbourhood frames}, derived from~\cite{Dal22}.
  We use a canonical model construction to prove completeness
  of extensions of $\IM$ with respect to suitable classes of such frames.
  This differs from the \emph{intuitionistic neighbourhood models} used
  in~\cite{Gro25-im}. While the latter were useful for bridging the
  gap between the first-order and modal perspectives of the logic,
  the former semantics provides us with valuable intuition for the
  sequent calculus.

  Next, we present a sound and complete calculus for $\IM$.
  Considering that restrictions of Gentzen-style sequent calculi yield
  constructive monotone modal logics~\cite{Dal22,Dal25},
  and by analogy with the nested sequent calculus for $\IK$~\cite{Str13}, 
  we now start from a \emph{structured} sequent calculus for classical $\EM$
  and adapt it to~$\IM$.
  We consider to this purpose the calculus $\HM$
  from~\cite{DalLelOliPim21}.
  Besides the comma and the sequent arrow, $\HM$ employs two additional structural connectives: 
  the hypersequent bar, which allows one to represent sets of sequents and to efficiently extract neighbourhood countermodels from failed derivations,
  and blocks of formulas, which intuitively represent neighbourhood sets
  associated with the corresponding sequent/worlds,
  and also allow for a modular extension of the base calculus
  by means of rules handling blocks.
Since countermodel extraction is outside the scope of the present work, we disregard here hypersequents
(which is possible to do by still preserving the syntactic completeness of the calculus)
and restrict to simple sequents endowed with blocks.
  Our main contributions are as follows:
\begin{enumerate}
  \item We provide a cut-free, terminating calculus $\CIM$ for $\IM$
        by restricting a suitable sequent calculus with blocks to single-succedent sequents.
  \item We show that this is closely related to a calculus for $\WM$:
        $\CIM$ is based on an operation of pruning that cuts out output formulas
        in the application of an implication rule, and changing this to
        a more radical form of pruning (that we call output thinning)
        provides a calculus for $\WM$.
  \item We prove that $\IM$ is decidable, and re-prove~\cite[Theorem~4.11]{Gro25-im} that $\IM$ and $\WM$ 
        have different $\Box$-free fragments
  \item We obtain calculi for extensions of $\IM$ by using block rules
        for extensions of $\EM$~\cite{DalLelOliPim21}.
\end{enumerate}

%%%%%%%%%%%%%%%%%%%%%%%%%%%%%%%%%%%%%%%%%%%%%%%%%%%%%%%%%%%%%%%%%%%%%%%%%%%%%%%%
\section{Intuitionistic monotone modal logic}\label{sec:IM}

  We briefly recall the logics $\WM$ and $\IM$ from~\cite{Dal22,Gro25-im}.
  We then give a canonical model construction for $\IM$ and some of its extensions
  to prove completeness with respect to constructive neighbourhood frames
  (as opposed to the \emph{intuitionistic} neighbourhood frames from~\cite{Gro25-im}).

%%%%%%%%%%%%%%%%%%%%%%%%%%%%%%%%%%%%%%%%%%%%%%%%%%%%%%%%%%%%%%%%%%%%%%%%%%%%%%%%
\subsection{The logics}\label{subsec:logic}

Let $\mc{L}$ be the language generated by the grammar
$
    \phi ::= p \mid \bot \mid \phi \land \phi \mid \phi \lor \phi \mid \phi \imp \phi \mid \Box \phi \mid \diam \phi,
$
  where $p$ ranges over some arbitrary but fixed set $\Prop$ of proposition letters.
  As usual, we abbreviate $\neg \phi := \phi \imp \bot$ and $\top := \neg\bot$.
Besides, the \emph{complexity} of a formula is recursively defined via
  $c(p) = c(\bot) = 0$,
  $c(\phi \wedge \psi) = c(\phi \vee \psi) = c(\phi \to \psi) = c(\phi) + c(\psi) + 1$,
  and $c(\Box\phi) = c(\Diamond\phi) = c(\phi) + 1$.

  Intuitionistic analogues of monotone modal logics
  can be described using the following axioms:
  \begin{equation*}
    \ax{$\phi \imp \psi$}
      \llab{$\monbox$}
      \uinf{$\Box \phi \imp \Box \psi$}
      \disp,
    \qquad
    \ax{$\phi \imp \psi$}
      \llab{$\mondiam$}
      \uinf{$\diam \phi \imp \diam \psi$}
      \disp,
    \qquad
      \axmnc\ (\Box \phi \land \diam \neg \phi) \imp \bot,
    \qquad
      \axIdiam\ (\Box\top \imp \diam \phi) \imp \diam \phi.
  \end{equation*}

\begin{definition}\label{def:logics}
  Let $\mc{R} = \{ R_1, \ldots, R_n \}$ be a set of axioms and rules.
  Then we write $\IPL \oplus \mc{R}$ or $\IPL \oplus R_1 \oplus \cdots \oplus R_n$
  for the logic over the language $\mc{L}$ obtained by extending an axiomatisation
  for intuitionistic propositional logic with the axioms and rules in $\mc{R}$,
  and closing it under modus ponens and uniform substitution.
  We can define \emph{constructive}, \emph{Wijesekera} and
  \emph{intuitionistic monotone modal logic} as follows:
  \begin{equation*}
    \CM := \IPL \oplus \monbox \oplus \mondiam, \qquad
    \WM := \CM \oplus \axmnc, \qquad
    \IM := \WM \oplus \axIdiam.
  \end{equation*}
\end{definition}

  We are also interested in extensions of $\IM$ with the axioms
  $\axN, \axP, \axT, \axD$ and $\axK$,
  which are well known in the classical modal logic literature.
  For example, in the context of coalition logic $\axN$ and $\axP$ embody the idea
  that coalitions can bring about anything that is already true, but cannot achieve the impossible (because classically $\axP$ is equivalent to $\neg\Box\bot$), and in an epistemic context $\axT$ states that only
  true things can be known.
  The axiom $\axK$ is of interest because it yields intuitionistic normal modal logics
  weaker than $\IK$.
  In this paper, we consider:  
  \begin{equation*}
    \begin{array}{llllll}
      \axN & \Box\top
        & \axP & \Diamond\top
        & \axD & \Box\phi \imp \Diamond\phi \\[.2em]
      \axTbox & \Box\phi \imp \phi
        & \axTdiam & \phi \imp \Diamond\phi
        & \axTbd   & (\axTbox \wedge \axTdiam) \\[.2em]
      \axKbox & \Box(\phi \imp \psi) \imp (\Box \phi \imp \Box \psi) \qquad
        & \axKdiam & \Box(\phi \to \psi) \to (\Diamond\phi \to \Diamond\psi) \qquad
        & \axKbd & (\axKbox \wedge \axKdiam)
    \end{array}
  \end{equation*}
  For any logic $\mathsf{L}$, we write $\Phi \vdash_{\mathsf{L}} \phi$ if there exist
  $\phi_1, \ldots, \phi_n \in \Phi$ such that $\vdash_{\mathsf{L}} (\phi_1 \wedge \cdots \wedge \phi_n) \to \phi$. Since in this paper we focus on the logic $\IM$, if $\Ax$ is a set of axioms we abbreviate
  $\Phi \vdash_{\IM \oplus \Ax} \phi$ to $\Phi \vdash_{\Ax} \phi$.
  
  We will also refer to the following modal axioms, which relate to the preceding ones
  via Lemma~\ref{lemma:inclusions}.
 \begin{equation*}
    \begin{array}{llllll}
      \axCbox & \Box\phi \land \Box\psi \to \Box(\phi \land \psi)\qquad
        & \axW & \Box\phi \land \diam \psi \imp \diam(\phi \land \psi)\qquad
        & \axPbox & \Box\bot \to \bot \\[.4em]
      \str &      \ax{$\phi \land \psi \imp \bot$}
                \uinf{$\Box\phi \land \diam\psi \imp \bot$}
                \disp
        & \RD &   \ax{$\phi\land\psi\imp\bot$}
                \uinf{$\Box\phi\land\Box\psi\imp\bot$}
                \disp
    \end{array}
  \end{equation*}

\begin{lemma}\label{lemma:eq axioms}\label{lemma:inclusions}
  Over $\CM$, $\axmnc$ is equivalent to $\str$.
  Over $\WM$, $\RD$ is derivable from $\axD$.
  Also, we have:
  \begin{align*}
    %% Row 1
      &\vdash_{\WM \oplus \axKbox} \axCbox
      &&\vdash_{\WM \oplus \axCbox} \axKbox
      &&\vdash_{\WM \oplus \axP} \axPbox
      &&\vdash_{\WM \oplus \axT} \axD
      &&\vdash_{\WM \oplus \axN \oplus \axD} \axP \\
    %% Row 2
      &&&\vdash_{\WM \oplus \axKdiam} \axW
      &&\vdash_{\WM \oplus \axN} \axIdiam
      &&\vdash_{\WM \oplus \axT} \axP
      &&\vdash_{\WM \oplus \axK \oplus \axP} \axD
  \end{align*}
  Furthermore, we have $\vdash_{\IM \oplus \axD} \axP$.
  Since $\IM$ is stronger than $\WM$, we may replace every occurrence of $\WM$ with $\IM$.
\end{lemma}

  Derivations can be found in~\cite{Dal22,Dal25}.
  In light of Lemma \ref{lemma:inclusions}, combinations of axioms
  in $\{ \axN, \axP, \axT, \axD, \axK \}$ give rise to 14 different systems, depicted
  below, where arrows denote inclusion of logical systems.
  Each of the systems arises from extending the corresponding system
  from~\cite{Dal22} with $\axIdiam$.

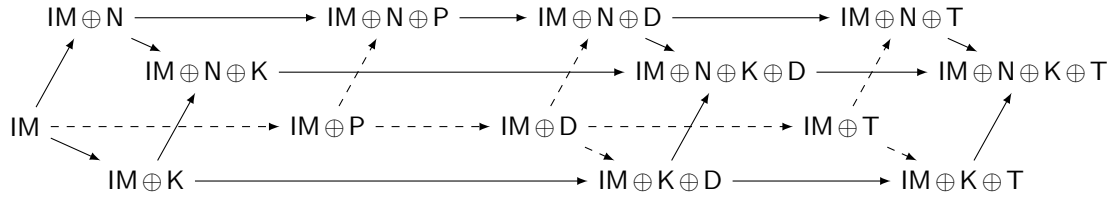
\begin{equation*}
\begin{tikzpicture}[>=latex, yscale=.9, font=\small]

\def\dx{4}       
\def\dxPD{2.8}    

\def\nx{0.8}
\def\ny{1.6}
\def\kx{1.6}
\def\ky{-0.8}

\def\xM{0}
\def\xP{\dx}
\def\xD{\dx+\dxPD}
\def\xT{\dx+\dxPD+\dx}

\node (M)  at (\xM,0) {$\IM$};
\node (MN) at ({\xM+\nx},{\ny}) {$\IM \oplus \axN$};
\node (MC) at ({\xM+\kx},{\ky}) {$\IM \oplus \axK$};
\node (K)  at ({\xM+\nx+\kx},{\ny+\ky}) {$\IM \oplus \axN \oplus \axK$};

\node (MP)  at (\xP,0) {$\IM \oplus \axP$};
\node (MNP) at ({\xP+\nx},{\ny}) {$\IM \oplus \axN \oplus \axP$};

\node (MD)  at (\xD,0) {$\IM \oplus \axD$};
\node (MND) at ({\xD+\nx},{\ny}) {$\IM \oplus \axN \oplus \axD$};
\node (MCD) at ({\xD+\kx},{\ky}) {$\IM \oplus \axK \oplus \axD$};
\node (KD)  at ({\xD+\nx+\kx},{\ny+\ky})
             {$\IM \oplus \axN \oplus \axK \oplus \axD$};

\node (MT)  at (\xT,0) {$\IM \oplus \axT$};
\node (MNT) at ({\xT+\nx},{\ny}) {$\IM \oplus \axN \oplus \axT$};
\node (MCT) at ({\xT+\kx},{\ky}) {$\IM \oplus \axK \oplus \axT$};
\node (KT)  at ({\xT+\nx+\kx},{\ny+\ky})
             {$\IM \oplus \axN \oplus \axK \oplus \axT$};

\draw[->] (M) -- (MN);
\draw[->] (M) -- (MC);
\draw[->] (MN) -- (K);
\draw[->] (MC) -- (K);

\draw[->, dashed] (MP) -- (MNP);

\draw[->, dashed] (MD) -- (MND);
\draw[->, dashed] (MD) -- (MCD);
\draw[->] (MND) -- (KD);
\draw[->] (MCD) -- (KD);

\draw[->, dashed] (MT) -- (MNT);
\draw[->, dashed] (MT) -- (MCT);
\draw[->] (MNT) -- (KT);
\draw[->] (MCT) -- (KT);

\draw[->, dashed] (M) -- (MP);
\draw[->, dashed] (MP) -- (MD);
\draw[->, dashed] (MD) -- (MT);

\draw[->] (MN) -- (MNP);
\draw[->] (MNP) -- (MND);
\draw[->] (MND) -- (MNT);

\draw[->] (MC) -- (MCD);
\draw[->] (MCD) -- (MCT);

\draw[->] (K) -- (KD);
\draw[->] (KD) -- (KT);

\end{tikzpicture}
\end{equation*}

%%%%%%%%%%%%%%%%%%%%%%%%%%%%%%%%%%%%%%%%%%%%%%%%%%%%%%%%%%%%%%%%%%%%%%%%%%%%%%%%
\subsection{Constructive neighbourhood semantics}
  
  Let $(X, \leq)$ be a preordered set. For $a \subseteq X$ we define
  ${\uparrow}a = \{ y \in X \mid x \leq y \text{ for some } x \in a \}$.
  Then $a$ is called an \emph{upset} if ${\uparrow}a = a$,
  and we write $\up(X, \leq)$ for the collection of upsets of $(X, \leq)$.

\begin{definition}
  A \emph{constructive neighbourhood frame} is a tuple $\mo{F} = (W, \leq, \neigh)$
  consisting of a set $W$, a preorder $\leq$ on $W$, and a
  function $\neigh : W \to \power(\power(W))$ called a \emph{neighbourhood function}.
  A \emph{constructive neighbourhood model} is a pair $\mo{M} = (\mo{F}, V)$
  comprising an constructive neighbourhood frame $\mo{F}$ and a
  valuation $V : \Prop \to \up(X, \leq)$ that assigns an upset of $(X, \leq)$
  to each proposition letter $p$.
  
  The interpretation of $\mc{L}$-formulas at a world $w$ of such a model $\mo{M}$
  is defined recursively by
  \begin{align*}
    \mo{M}, w \Vdash p &\iff w \in V(p) \\
    \mo{M}, w \Vdash \bot &\phantom{\iff}\text{never} \\
    \mo{M}, w \Vdash \phi \wedge \psi &\iff \mo{M}, w \Vdash \phi \text{ and } \mo{M}, w \Vdash \psi \\
    \mo{M}, w \Vdash \phi \vee \psi &\iff \mo{M}, w \Vdash \phi \text{ or } \mo{M}, w \Vdash \psi \\
    \mo{M}, w \Vdash \phi \to \psi
      &\iff \text{for all }v \geq w, \; \mo{M}, v \Vdash \phi \text{ implies } \mo{M}, v \Vdash \psi \\
    \mo{M}, w \Vdash \Box\phi
      &\iff \text{for all } v \geq w
            \text{ there exists } a \in \neigh(v)
            \text{ such that for all } u \in a
            \text{ we have } \mo{M}, u \Vdash \phi \\
    \mo{M}, w \Vdash \Diamond\phi
      &\iff \text{for all } v \geq w
            \text{ and all } a \in \neigh(v)
            \text{ there exists } u \in a
            \text{ such that } \mo{M}, u \Vdash \phi
  \end{align*}
  
  Let $\Phi \cup \{ \psi \} \subseteq \lan$.
  A constructive neighbourhood frame $\mo{F}$ \emph{validates} $\psi$
  if $(\mo{F}, V), w \Vdash \psi$ for every valuation $V$ for $\mo{F}$
  and every world $w$ in $\mo{F}$,
  and it \emph{validates} $\Phi$ if it validates every
  formula in $\Phi$, notation: $\mo{F} \Vdash \psi$ and $\mo{F} \Vdash \Phi$.
  If $\Ax$ is a set of axioms and $\phi \in \lan$,
  then we write $\Vdash_{\Ax} \phi$ if every constructive neighbourhood
  frame that validates the formulas in $\Ax$ also validates $\phi$.
\end{definition}

  It was shown in~\cite{Dal22} that constructive neighbourhood frames provide sound and complete
  semantics with respect to various extensions of~$\WM$.
  They can also be used for (extensions of)~$\IM$.

\begin{lemma}\label{lem:frame-conditions}
  Let $\mo{F} = (W, \leq, \neigh)$ be a constructive neighbourhood frame.
  \begin{enumerate}
    \item If $\neigh(w) \neq \emptyset$ for all $w \in W$,
          then $\mo{F}$ validates $\axN$;
    \item If $\emptyset \notin \neigh(w)$ for all $w \in W$,
          then $\mo{F}$ validates $\axP$;
    \item If $w \in U$ for all $U \in \neigh(w)$ and $w \in W$,
          then $\mo{F}$ validates $\axTbd$;
    \item If $U_1 \cap U_2 \neq \emptyset$ for all $U_1, U_2 \in \neigh(w)$
          and $w \in W$, then $\mo{F}$ validates $\axD$;
    \item If $\neigh(w)$ is closed under binary intersections for all $w \in W$,
          then $\mo{F}$ validates $\axKbd$;
    \item \label{it:full}
          If ($\neigh(w) \neq \emptyset$ and $w \leq v$) implies $\neigh(v) \neq \emptyset$,
          for all $w, v \in W$, then $\mo{F}$ validates $\axIdiam$.
  \end{enumerate}
\end{lemma}
\begin{proof}
  The first five items can be found in~\cite{Dal22}.
  For the last one, suppose $\mo{F}$ satisfies the given condition,
  and let $\mo{M} = (\mo{F}, V)$ be a model based on $\mo{F}$.
  Let $w \in W$ and suppose $w \Vdash \Box\top \to \Diamond\phi$.
  We show that any neighbourhood $U$ of any $v \geq w$ contains an element
  that satisfies $\phi$. Let $v \geq w$. Then $v \Vdash \Box\top \to \Diamond\phi$.
  If $\neigh(v) \neq \emptyset$, then by the frame condition we have
  $v \Vdash \Box\top$, so that $v \Vdash \Diamond\phi$, which implies that
  every neighbourhood of $v$ contains an element satisfying $\phi$.
  If, on the other hand, $\neigh(v) = \emptyset$, then it is trivially true
  that every neighbourhood of $v$ contains a world satisfying $\phi$.
  Therefore $w \Vdash \Diamond\phi$. Since $w$ was arbitrary, this proves
  that the frame validates $\axIdiam$.
\end{proof}

\begin{definition}\label{def:continual}
  Constructive neighbourhood frames that satisfy
  Item~\eqref{it:full} above are called \emph{continual}.
\end{definition}

%%%%%%%%%%%%%%%%%%%%%%%%%%%%%%%%%%%%%%%%%%%%%%%%%%%%%%%%%%%%%%%%%%%%%%%%%%%%%%%%
\subsection{Canonical model constructions}

  Throughout this subsection, let
  $\Ax \subseteq \{ \axN, \axP, \axTbd, \axD, \axKbd \}$.
  We prove that $\IM \oplus \Ax$ is sound and complete with
  respect to a suitable class of continual constructive neighbourhood frames.
  Akin to~\cite{Dal22,Wij90}, we use so-called segments, rather than prime theories,
  to construct a canonical model.
  The case for $\IM$ without additional axioms was proven (indirectly)
  in~\cite[Theorem~4.8]{Gro25-im}.
  
\begin{definition}
  An \emph{$\Ax$-prime theory} is a subset $\Phi \subseteq \lan$ that is
  \emph{consistent} ($\Phi \not\vdash_{\Ax} \bot$),
  \emph{deductively closed} ($\Phi \vdash_{\Ax} \phi$ implies $\phi \in \Phi$)
  and \emph{prime} ($\phi \vee \psi \in \Phi$ implies $\phi \in \Phi$ or $\psi \in \Phi$).
\end{definition}

  We write $\PF_{\Ax}$ for the set of $\Ax$-prime theories.
  We can prove the Lindenbaum lemma as usual:

\begin{lemma}
  Let $\Phi \cup \{ \psi \} \subseteq \lan$ be such that $\Phi \not\vdash_{\Ax} \psi$.
  Then there exists an $\Ax$-prime theory $\Phi'$ such that
  $\Phi \subseteq \Phi'$ and $\psi \notin \Phi'$.
\end{lemma}
  
  Our canonical model is not built from prime theories, but rather from
  pairs consisting of a prime theory and a family of sets of prime theories
  that encodes the neighbourhoods of the world.

\begin{definition}
  An \emph{$\Ax$-segment} is a pair $(\Phi, \ms{N})$ consisting of an
  $\Ax$-prime theory $\Phi$ and a collection $\ms{N}$ of sets of $\Ax$-prime theories,
  such that:
  \begin{itemize}
    \item If $\Box\phi \in \Phi$ then there exists a $U \in \ms{N}$ such that
          $\Psi \in U$ implies $\phi \in \Psi$;
    \item If $\Diamond\phi \in \Phi$ then for all $U \in \ms{N}$ there exists
          some $\Psi \in U$ such that $\phi \in \Psi$.
  \end{itemize}
\end{definition}

  Next, we consider several particular segments based on a given prime
  theory $\Phi$. This also shows that we can always extend a prime
  theory to a segment.
  For a formula $\phi$, define $\tilde{\phi} := \{ \Phi \in \PF_{\Ax} \mid \phi \in \Phi \}$ and $\tilde{\phi}^c := \PF_{\Ax} \setminus \tilde{\phi}$.
  (When using the notation $\tilde{\phi}$, the set $\Ax$ will be clear from context.)

\begin{lemma}\label{lem:seg-BD}
  Let $\Phi$ be an $\Ax$-prime theory and $\psi \in \lan$ such that
  $\Diamond\psi \notin \Phi$.
  Define
  \begin{equation*}
    \ms{B}_{\Phi} := \{ \tilde{\phi} \mid \Box\phi \in \Phi \}
    \quad\text{and}\quad
    \ms{D}_{\Phi,\psi} :=
      \begin{cases}
        \ms{B}_{\Phi} \cup \{ \tilde{\psi}^c \} &\text{if $\axKbd \notin \Ax$} \\
        \ms{B}_{\Phi} \cup \{ \tilde{\psi}^c \} \cup \{ \tilde{\phi} \cap \tilde{\psi}^c \mid \Box\phi \in \Phi \} &\text{if $\axKbd \in \Ax$}
      \end{cases}
  \end{equation*}
  Then $(\Phi, \ms{B}_{\Phi})$ and $(\Phi, \ms{D}_{\Phi,\psi})$ are $\Ax$-segments.
\end{lemma}

\begin{proof}
  In each of the cases, the $\Box$-condition is clearly satisfied,
  so we prove the $\Diamond$-conditions.
  For $(\Phi, \ms{B}_{\Phi})$, we need to show that for all
  $\Box\phi, \Diamond\chi \in \Phi$ we have
  $\tilde{\phi} \cap \tilde{\chi} \neq \emptyset$.
  Suppose towards a contradiction that this were not the case,
  then $\phi, \chi \vdash_{\Ax} \bot$, hence by $\str$
  $\Box\phi, \Diamond\chi \vdash_{\Ax} \bot$.
  But then $\bot \in \Phi$, a contradiction.
  
  For $(\Phi, \ms{D}_{\Phi,\psi})$, if $\axKbd \notin \Ax$
  then we need to additionally show that
  $\tilde{\psi}^c \cap \tilde{\chi} \neq \emptyset$ for any $\Diamond\chi \in \Phi$.
  If $\tilde{\psi}^c \cap \tilde{\chi} = \emptyset$, then
  $\tilde{\chi} \subseteq \tilde{\psi}$, which implies $\chi \vdash_{\Ax} \psi$.
  But then $\mondiam$ entails $\Diamond\chi \vdash_{\Ax} \Diamond\psi$,
  so that $\Diamond\psi \in \Phi$ by deductive closure of $\Phi$, a contradiction.
  
  Finally, if $\axKbd \in \Ax$ then we need to also prove that
  $\Box\phi, \Diamond\chi \in \Phi$ we have
  $(\tilde{\phi} \cap \tilde{\psi}^c) \cap \tilde{\chi} \neq \emptyset$.
  Suppose towards a contradiction that this intersection is empty.
  Then $\tilde{\phi} \cap \tilde{\chi} \subseteq \tilde{\psi}$,
  hence $\phi, \chi \vdash_{\Ax} \psi$,
  so by $\axKdiam$ also $\Box\phi, \Diamond\chi \vdash_{\Ax} \Diamond\psi$.
  This implies $\Diamond\psi \in \Phi$, a contradiction.
\end{proof}

\begin{lemma}\label{lem:segments}
  Let $(\Phi, \ms{N})$ be an $\Ax$-segment and $\psi$ a formula with $\Diamond\psi \notin \Phi$, such that
  $\ms{N} \in \{ \ms{B}_{\Phi}, \ms{D}_{\Phi,\psi} \}$.
  
  \medskip\noindent\hspace{.5em}
  \begin{tabular}{p{.25em}lp{.25em}l}
    1. & If $\axN \in \Ax$ then $\ms{N} \neq \emptyset$;
       & 4. & $\axD \in \Ax$ then $U_1 \cap U_2 \neq \emptyset$ for all $U_1, U_2 \in \ms{N}$; \\[.4em]
    2. & If $\axP \in \Ax$ then $\emptyset \notin \ms{N}$;
       & 5. & If $\axKbd \in \Ax$ then $U_1 \cap U_2 \in \ms{N}$ for all $U_1, U_2 \in \ms{N}$. \\[.4em]
    3. & If $\axTbd \in \Ax$, then $\Phi \in U$ for all $U \in \ms{N}$;
  \end{tabular}
\end{lemma}

\begin{proof}
  (1) Since $\Box\top \in \Phi$ we always get $\PF_{\Ax} \in \ms{N}$.

  (2) By assumption $\Diamond\top \in \Phi$, hence it follows from
      $\axmnc$ that $\Box\bot \notin \Phi$.
      This entails $\emptyset \notin \ms{B}_{\Phi}$.
      If $\emptyset \in \ms{D}_{\Phi,\psi}$ then we must have $\tilde{\psi}^c = \emptyset$
      (if $\axK \notin \Ax$) or $\tilde{\phi} \cap \tilde{\psi}^c = \emptyset$
      for some $\Box\phi \in \Phi$ (if $\axK \in \Ax$).
      In the former case $\psi$ must be equivalent to $\top$,
      but then the $\Diamond\top \in \Phi$ implies $\Diamond\psi \in \Phi$,
      a contradiction.
      In the latter case, combining $\axP$ and $\axKdiam$ yields
      $\axD$, so $\Box\phi \in \Phi$ implies $\Diamond\phi \in \Phi$.
      Moreover, the assumption $\tilde{\phi} \cap \tilde{\psi}^c = \emptyset$
      yields $\phi \vdash_{\Ax} \psi$, hence by monotonicity
      $\Diamond\phi \vdash_{\Ax} \Diamond\psi$ and therefore $\Diamond\psi \in \Phi$,
      and we have reached a contradiction again.
      Therefore we conclude $\emptyset \notin \ms{D}_{\Phi,\psi}$, as desired..

  (3) If $U = \tilde{\phi}$ for some $\Box\phi \in \Phi$,
      then using $\axTbox$ we find $\phi \in \Phi$, hence $\Phi \in \tilde{\phi} = U$.
      If $U = \tilde{\psi}^c$ for some $\Diamond\psi \notin \Phi$
      then using $\axTdiam$ we find $\psi \notin \Phi$, hence
      $\Phi \in \tilde{\psi}^c = U$.
      Finally, if $U$ is of the form $\tilde{\phi} \cap \tilde{\psi}^c$
      for some $\Box\phi \in \Phi$ and $\Diamond\psi \notin \Phi$,
      then combining the previous two cases gives $\Phi \in U$ again.
      
  (4) Suppose $U_1 = \tilde{\phi}$ and $U_2 = \tilde{\chi}$ for some
      $\Box\phi, \Box\chi \in \Phi$.
      Then $\axD$ entails $\Diamond\chi \in \Phi$, hence
      $\Box\phi, \Diamond\chi \not\vdash_{\Ax} \bot$.
      Now $\str$ entails $\phi, \chi \not\vdash_{\Ax} \bot$,
      so that the Lindenbaum lemma tells us that
      $\tilde{\phi} \cap \tilde{\chi} \neq \emptyset$, as desired.
      
      Next, suppose $U_1 = \tilde{\phi}$ and $U_2 = \tilde{\psi}^c$
      for some $\Box\phi \in \Phi$ and $\Diamond\psi \notin \Phi$.
      If $\tilde{\phi} \cap \tilde{\psi}^c = \emptyset$
      then $\tilde{\phi} \subseteq \tilde{\psi}$, so $\phi \vdash_{\Ax} \psi$.
      Now $\monbox$ yields $\Box\phi \vdash_{\Ax} \Box\psi$
      and by $\axD$ we have $\Box\psi \vdash_{\Ax} \Diamond\psi$,
      so that deductive closure of $\Phi$ entails $\Diamond\psi \in \Phi$,
      a contradiction.
      If $\axKbd \notin \Ax$ then this completes the proof for
      $(\Phi, \ms{D}_{\Phi,\psi})$.
      If $\axKbd \in \Ax$ then $\axCbox$ entails that we can reduce
      any intersection of the sets in $\ms{D}_{\Phi,\psi}$
      to a set of the form $\tilde{\phi} \cap \tilde{\chi}$ or
      $\tilde{\phi} \cap \tilde{\psi}^c$, where $\Box\phi, \Box\chi \in \Phi$
      and $\Diamond\psi \notin \Phi$, and we have already seen that
      such intersections are nonempty.

  (5) Both cases follow immediately from the fact that
      derivability of $\axCbox$ entails that $\ms{B}_{\Phi}$ is closed
      under binary intersections.
\end{proof}

\begin{lemma}\label{lem:truth-lemma-helper}
  Let $\Phi$ be an $\Ax$-prime theory such that
  $\Box\top \notin \Phi$ and $\Diamond\phi \notin \Phi$.
  Then there exists an $\Ax$-prime theory $\Phi'$ extending $\Phi$
  that contains $\Box\top$ but not $\Diamond\phi$.
\end{lemma}
\begin{proof}
  It suffices to prove $\Phi, \Box\top \not\vdash_{\Ax} \Diamond\phi$,
  because then the Lindenbaum lemma yields the desired $\Phi'$.
  If this is not the case, i.e.~if $\Phi, \Box\top \vdash_{\Ax} \Diamond\phi$,
  then there exists $\gamma \in \Phi$ such that
  $\gamma, \Box\top \vdash_{\Ax} \Diamond\phi$,
  hence $\gamma \vdash_{\Ax} \Box\top \to \Diamond\phi$.
  Then $\axIdiam$ entails $\gamma \vdash \Diamond\phi$,
  so that $\Diamond\phi \in \Phi$ by deductive closure of $\Phi$,
  a contradiction.
\end{proof}

  We base our canonical model on segments arising from Lemma~\ref{lem:seg-BD}.
  To ensure that the canonical model is continual (Definition~\ref{def:continual}),
  we impose the restriction that if $\Box\top \notin \Phi$ then
  we only use the segment $(\Phi, \ms{B}_{\Phi})$.

\begin{definition}
  The \emph{canonical model} is the tuple 
  $\mo{M}_{\Ax} = (W_{\Ax}, \subsetsim, \neigh_{\Ax}, V_{\Ax})$, where:
  \begin{align*}
    W_{\Ax} &:= \{ (\Phi, \ms{B}_{\Phi}) \mid \Phi \in \PF_{\Ax} \}
                \cup \{ (\Phi, \ms{D}_{\Phi,\psi}) 
                        \mid \Phi \in \PF_{\Ax}, 
                             \Box\top \in \Phi, 
                             \Diamond\psi \notin \Phi \} \\
    (\Phi, \ms{N}) &\subsetsim (\Psi, \ms{M}) \iff \Phi \subseteq \Psi \\
    \neigh_{\Ax}(\Phi, \ms{N})
      &:= \{ \{ (\Psi, \ms{M}) \in W_{\Ax} \mid \Psi \in \ms{U} \} \mid \ms{U} \in \ms{N} \} \\
    V_{\Ax}(p) &:= \{ (\Phi, \ms{N}) \in W_{\Ax} \mid p \in \Phi \}
  \end{align*}
\end{definition}

  Note that the definition above defines two different canonical
  models, depending on whether or not $\axKbd \in \Ax$ (see Lemma~\ref{lem:seg-BD}).
  Both are continual:
  if $(\Phi, \ms{N})$ is a world such that $\neigh(\Phi, \ms{N}) = \emptyset$,
  then $\ms{N} = \emptyset$ and we must have $\Box\top \notin \Phi$.
  This implies $\ms{N} = \ms{B}_{\Phi}$.
  Then for any $(\Psi, \ms{M}) \subsetsim (\Phi, \ms{N})$ we find
  $\ms{M} = \ms{B}_{\Psi} = \emptyset$.

\begin{lemma}[Truth lemma] \label{lem:truth}
  For all $\phi \in \lan$ and $(\Phi, \ms{N}) \in W_{\Ax}$ we have
  $\mo{M}_{\Ax}, (\Phi, \ms{N}) \Vdash \phi$ iff $\phi \in \Phi$.
\end{lemma}
\begin{proof}
  The proof proceeds by induction on the structure of $\phi$.
  The propositional cases are routine.
  
  \medskip\noindent
  \textit{Case $\phi = \Box\psi$.}
    By definition of a segment, $\Box\psi \in \Phi$ implies
    $(\Phi, \ms{N}) \Vdash \Box\psi$.
    Conversely, if $(\Phi, \ms{N}) \Vdash \Box\psi$
    then since $(\Phi, \ms{N}) \subsetsim (\Phi, \ms{B}_{\Phi})$
    we have $(\Phi, \ms{B}_{\Phi}) \Vdash \Box\psi$.
    This means that there exists some $\Box\chi \in \Phi$
    such that every prime theory in $\tilde{\chi}$ satisfies $\psi$,
    hence by induction every such prime theory contains $\psi$,
    i.e.~$\tilde{\chi} \subseteq \tilde{\psi}$.
    Therefore $\chi \vdash \psi$, so combining monotonicity and the fact
    that $\Box\chi \in \Phi$ entails $\Box\psi \in \Phi$.
  
  \medskip\noindent
  \textit{Case $\phi = \Diamond\psi$.}
    By definition $\Diamond\psi \in \Phi$ implies
    $(\Phi, \ms{N}) \Vdash \Diamond\psi$.
    For the converse, suppose $\Diamond\psi \notin \Phi$.
    If $\Box\top \in \Phi$, then
    $(\Phi, \ms{D}_{\Phi,\psi})$ is a segment in $W_{\Ax}$
    such that $\tilde{\psi}^c \in \ms{D}_{\Phi,\psi}$,
    hence $\{ (\Psi, \ms{M}) \mid \Psi \in \tilde{\psi}^c \} \in \neigh_{\Ax}(\Phi, \ms{D}_{\Phi,\psi})$.
    Since $\psi \notin \Psi$ for all such $(\Psi, \ms{M})$,
    by the induction hypothesis no segment in 
    $\{ (\Psi, \ms{M}) \mid \Psi \in \tilde{\psi}^c \}$ satisfies $\psi$.
    Since $(\Phi, \ms{N}) \subsetsim (\Phi, \ms{D}_{\Phi,\psi})$
    this proves $(\Phi, \ms{N}) \not\Vdash \Diamond\psi$.
    If $\Box\top \notin \Phi$ then
    Lemma~\ref{lem:truth-lemma-helper} yields a prime theory $\Phi'$
    containing $\Phi$ and $\Box\top$ but not $\Diamond\psi$,
    so that $(\Phi, \ms{N}) \subsetsim (\Phi', \ms{D}_{\Phi',\psi})$
    witnesses $(\Phi, \ms{N}) \not\Vdash \Diamond\psi$.
\end{proof}

  In the usual way, we can now obtain soundness and completeness:

\begin{theorem}
  Let $\Ax \subseteq \{ \axN, \axP, \axTbd, \axD, \axKbd \}$.
  Then $\Vdash_{\IM \oplus \Ax} \phi$ if and only if $\vdash_{\IM \oplus \Ax} \phi$.
\end{theorem}

%%%%%%%%%%%%%%%%%%%%%%%%%%%%%%%%%%%%%%%%%%%%%%%%%%%%%%%%%%%%%%%%%%%%%%%%%%%%%%%%
\section{A sequent calculus for $\IM$}\label{sec:CIM}

  We present the sequent calculus $\CIM$,
  provide a syntactic proof of cut elimination,
  show that it is sound and complete with respect to $\IM$,
  and observe that a small modification yields a calculus for $\WM$.

%%%%%%%%%%%%%%%%%%%%%%%%%%%%%%%%%%%%%%%%%%%%%%%%%%%%%%%%%%%%%%%%%%%%%%%%%%%%%%%%
\subsection{The sequent calculus $\CIM$}

\begin{definition} \
  \begin{enumerate}
    \item A \emph{marked formula} is a formula $\phi \in \lan$ marked with the
          \emph{input polarity} $\bullet$ or the \emph{output polarity} $\circ$.
          Marked formulas $\bl{\phi}$ and $\wh{\phi}$ are called \emph{input} and
          \emph{output formulas}, respectively.
          We write $\MForm$ for the collection of marked formulas.
    \item A \emph{block} $\B$ is a structure of the form $\nest{\D}$, 
          where $\D$ is a set of formulas that contains one or zero formulas.
          We write $|\Delta|$ for the number of formulas in $\Delta$.
	    A block $\nest{\D}$ is an \emph{input block}, denoted $\bl{\B}$, if $\D$ contains 
          an input formula or is empty, and is an \emph{output block}
          if $\D$ contains an output formula.
          The set of blocks is denoted by $\Blocks$.
    \item A \emph{sequent} $\G$ is defined by the following grammar, 
          where $n, k \geq 0$,
          and $\Lambda$ is understood as a multiset:
          \begin{equation*}
            \G ::= \Lambda, \Pi
            \qquad
            \Lambda ::= \phi^\bullet_1, \ldots, \phi^\bullet_n, \B^\bullet_1, \ldots, \B^\bullet_k
            \qquad
            \Pi ::= \emptyset \mid \wh{\phi} \mid \nest{\wh{\phi}}.
          \end{equation*}
  \end{enumerate}
\end{definition}

In other words, we consider (single-succedent) sequents
with at most one output formula, 
that can occur inside or outside a block.
We use $\G, \D, \Sigma$ as variables for sequents, 
$\Lambda$ as variable for input sequents (with only input components),
and $\Pi$ as variable for output sequents (with only output components).

\begin{definition}
For a sequent $\G$, we define its \emph{output pruning} $\pr{\G}$ to be the same sequent
with the output formula removed (if present).
\end{definition}

Note that
if the output formula occurs inside a block,
the block is not deleted but only emptied.
For instance, $\pr{(\bl{\phi}, \nest{\bl{\psi}}, \wh{\chi})} = \pr{(\bl{\phi}, \nest{\bl{\psi}})}
= \bl{\phi}, \nest{\bl{\psi}}$,
and $\pr{(\bl{\phi}, \nest{\bl{\psi}}, \nest{\wh{\chi}})} = \bl{\phi}, \nest{\bl{\psi}}, \enest$,
where in the latter pruning the output block $\nest{\wh{\chi}}$ is turned into the input block $\enest$. %Tiz:changed
  The \emph{formula interpretation}
  $\fint$ maps sequents to formulas in $\lan$:

\begin{definition}\label{def:fint}
  Define $\fintf : \MForm \cup \Blocks \cup \{ \emptyset \} \to \lan$ by
  $\fintf(\bl{\phi}) = \fintf(\wh{\phi}) = \phi$ and 
  \begin{equation*}
    \fintf(\nest{ \ }) = \Box \top, \quad
    \fintf(\nest{\bl{\phi}}) = \Box \phi, \quad
    \fintf(\nest{\wh{\phi}}) = \diam \phi \quad\text{and}\quad
    \fintf(\emptyset) = \bot.
  \end{equation*}
  Then the \emph{formula interpretation} of a sequent $\Gamma = \Lambda, \Pi$
  is given by
  $\fint(\Gamma) := \big( \bigwedge_{x \in \Lambda} \fintf(x) \big) \imp \fintf(\Pi)$,
  where we take the empty meet to be $\top$.
\end{definition}

  For instance, $\fint(\bl{\phi}, \nest{\bl{\psi}}, \nest{}) = \phi \wedge \Box \psi \wedge \Box\top \imp \bot$ and
  $\fint(\nest{\bl{\phi}}, \nest{\wh{\phi}}) = \Box \phi \imp \diam \phi$.

\begin{definition}
The calculus $\CIM$ for $\IM$ is defined by the rules in Figure~\ref{fig:CIM}.
\end{definition}

The calculus is built on a standard G3-style single-succedent calculus for $\IPL$
\cite{TroSch00},
and contains input and output rules for all connectives.
Ignoring the hypersequent structure (which is not needed for completeness), $\CIM$ can be seen as a single-succedent restriction of the calculus $\HM$ for the classical monotone modal classical logic $\EM$ defined in~\cite{DalLelOliPim21}.
Note also that contrary to nested calculi~\cite{AriDasStr15,Str13}, the blocks of $\CIM$ can only contain (sets of) formulas, and therefore cannot be nested.

\begin{figure}[ht]
\centering
\fbox{\begin{minipage}{\textwidth}
\vspace{0.1cm}
\centering
\ax{\phantom{$\Gamma$}}
\llab{$\bbot$}
\uinf{$\G, \bl{\bot}$}
\disp
\quad
\ax{\phantom{$\Gamma$}}
\llab{$\id$}
\uinf{$\Lambda, \bl{p}, \wh{p}$}
\disp
\quad
\ax{$\G, \bl{\phi}, \bl{\psi}$}
\llab{$\bland$}
\uinf{$\G, \bl{\phi \land \psi}$}
\disp
\quad
\ax{$\Lambda, \wh{\phi}$}
\ax{$\Lambda, \wh{\psi}$}
\llab{$\wland$}
\binf{$\Lambda, \wh{\phi \land \psi}$}
\disp

\vspace{0.3cm}
\ax{$\G, \bl{\phi}$}
\ax{$\G, \bl{\psi}$}
\llab{$\blor$}
\binf{$\G, \bl{\phi \lor \psi}$}
\disp
\quad
\ax{$\Lambda, \wh{\phi}$}
\llab{$\wlor$}
\uinf{$\Lambda, \wh{\phi \lor \psi}$}
\disp
\quad
\ax{$\Lambda, \wh{\psi}$}
\llab{$\wlor$}
\uinf{$\Lambda, \wh{\phi \lor \psi}$}
\disp

\vspace{0.3cm}
\ax{$\pr{\G}, \bl{\phi \imp \psi}, \wh{\phi}$}
\ax{$\G, \bl{\psi}$}
\llab{$\bimp$}
\binf{$\G, \bl{\phi \imp \psi}$}
\disp
\quad
\ax{$\Lambda, \bl{\phi}, \wh{\psi}$}
\llab{$\wimp$}
\uinf{$\Lambda, \wh{\phi \imp \psi}$}
\disp

\vspace{0.3cm}
\ax{$\G, \nest{\bl{\phi}}$}
\llab{$\bbox$}
\uinf{$\G, \bl{\Box \phi}$}
\disp
\quad
\ax{$\D, \wh{\phi}$}
\llab{$\wbox$}
\rlab{($|\D| \leq 1$)}
\uinf{$\Lambda, \nest{\D}, \wh{\Box \phi}$}
\disp
\quad
\ax{$\bl{\phi}, \D$}
\llab{$\bdiam$}
\rlab{($|\D| \leq 1$)}
\uinf{$\G, \bl{\diam \phi}, \nest{\D}$}
\disp
\quad
\ax{$\Lambda, \nest{\wh{\phi}}$}
\llab{$\wdiam$}
\uinf{$\Lambda, \wh{\diam \phi}$}
\disp
\end{minipage}}
\caption{\label{fig:CIM}The rules of $\CIM$.}
\end{figure}

\begin{example}\label{exm:der}
  The formulas $\wh{(\axIdiam)}$ and $\wh{((\diam\bot \imp \bot) \imp \bot) \imp\diam\bot}$
  can be derived as follows:
  \begin{equation*}
\begin{small}
\ax{$\wh{\top}$}
\llab{$\wbox$}
\uinf{$\bl{\Box\top \imp \diam \phi}, \nest{\ }, \wh{\Box\top}$}
\ax{$\bl{\phi}, \wh{\phi}$}
\rlab{$\bdiam$}
\uinf{$\bl{\diam \phi}, \nest{\wh{\phi}}$}
\rlab{$\bimp$}
\binf{$\bl{\Box\top \imp \diam \phi}, \nest{\wh{\phi}}$}
\rlab{$\wdiam$}
\uinf{$\bl{\Box\top \imp \diam \phi}, \wh{\diam \phi}$}
\rlab{$\wimp$}
\uinf{$\wh{(\Box\top \imp \diam \phi) \imp \diam \phi}$}
\disp
\quad
\ax{}
\llab{$\bbot$}
\uinf{$\bl{\bot}, \nest{\wh{\bot}}$}
\ax{}
\rlab{$\bbot$}
\uinf{$\bl{\bot}$}
\rlab{$\bdiam$}
\uinf{$\bl{(\diam\bot \imp \bot) \imp \bot}, \enest, \bl{\diam\bot}, \wh{\bot}$}
\rlab{$\wimp$}
\uinf{$\bl{(\diam\bot \imp \bot) \imp \bot}, \enest, \wh{\diam\bot \imp \bot}$}
\rlab{$\bimp$}
\binf{$\bl{(\diam\bot \imp \bot) \imp \bot}, \nest{\wh{\bot}}$}
\rlab{$\wdiam$}
\uinf{$\bl{(\diam\bot \imp \bot) \imp \bot}, \wh{\diam\bot}$}
\rlab{$\wimp$}
\uinf{$\wh{((\diam\bot \imp \bot) \imp \bot) \imp\diam\bot}$}
\disp
\end{small}
\end{equation*}
\end{example}

We prove that $\CIM$ and $\IM$ are equivalent by showing mutual derivability of rules of $\CIM$ and axioms of $\IM$.
Note that the principal block $\nest{\D}$ of the $\bdiam$-rule can be either an input or output block,
so that the rule simultaneously expresses $\str$ and $\mondiam$ of $\IM$. 
Moreover, observe the crucial role of output pruning $\downarrow$ in the derivation of $\axIdiam$.
Conversely, we show that the rule $\bimp$ with pruning $\downarrow$
is only derivable in presence of~$\axIdiam$.
The output pruning $\downarrow$ can be seen as the distinctive feature of the calculus $\CIM$.
We will see in Section~\ref{sec:CWM} that a slight modification of this notion of pruning
yields a calculus for $\WM$.

\begin{remark}\label{rem:sem}
  For a semantic intuition of $\CIM$,
  sequents are interpreted as worlds of a constructive neighbourhood model,
  input formulas are taken to be true and output formulas to be false at
  the corresponding world, and blocks are viewed as neighbourhoods associated
  with that world supporting or falsifying the contained input or output formulas, respectively. 
  Interpreting backward applications of $\bimp$ as transitioning
  from one world to a $\less$-successor, the preservation of (emptied) output
  blocks via $\downarrow$ corresponds to the property of being continual
  that characterises constructive neighbourhood models for~$\IM$.
\end{remark}

A rule $\G_1, \ldots, \G_n / \G$
is \emph{sound} with respect to $\IM$
if $\vdIM \fint(\G_j)$ for all $j \in \{ 1, \ldots, n \}$ implies $\vdIM \fint(\G)$.

\begin{proposition}\label{prop:sound-rules}
All rules of $\CIM$ are sound with respect to $\IM$.
\end{proposition}
\begin{proof}
All propositional rules except $\bimp$ are standard.
Moreover in the rules $\bbox$ and $\wdiam$
the formula interpretations of the premiss and the conclusion coincide, so
they are trivially sound. 

\medskip
\noindent
\textit{Case for rule $\bimp$.} \; Let $\G = \Lambda, \Pi$.
  If $\Pi = \wh{\chi}$ or $\Pi = \emptyset$, 
  then the rule is sound by standard propositional reasoning.
  If not, then we must have $\Pi = \nest{\wh{\chi}}$, so that the rule has premisses
($\Lambda, \bl{\phi \imp \psi}, \nest{ \ }, \wh{\phi}$) and
($\Lambda, \bl{\psi}, \nest{\wh{\chi}}$), and conclusion
($\Lambda, \bl{\phi \imp \psi}, \nest{\wh{\chi}}$).
If the formula interpretations of the premisses hold, then
$\vdIM \fint(\Lambda, \bl{\phi \imp \psi}, \nest{ \ }, \wh{\phi}) = \fint(\Lambda) \land (\phi \imp \psi) \land \Box\top \imp \phi$
and $\vdIM \fint(\Lambda, \bl{\psi}, \nest{\wh{\chi}}) = \fint(\Lambda) \land \psi \imp \diam \chi$.
This implies $\vdIM \fint(\Lambda) \land (\phi \imp \psi) \land \Box\top \imp \diam \chi$,
which can be curried to $\vdIM \fint(\Lambda) \land (\phi \imp \psi) \imp (\Box\top \imp \diam \chi)$.
Now axiom $\axIdiam$ yields
$\vdIM \fint(\Lambda) \land (\phi \imp \psi) \imp \diam \chi = \fint(\Lambda, \bl{\phi \imp \psi}, \nest{\wh{\chi}})$, as desired

\medskip
\noindent
\textit{Case for rule $\wbox$.} \; Either $\D = \bl{\psi}$ for some $\psi$
or $\D = \emptyset$. 
If $\D = \bl{\psi}$, then from $\vdIM \fint(\bl{\psi}, \wh{\phi}) = \psi \imp \phi$,
by $\monbox$ we obtain $\vdIM \Box \psi \imp \Box \phi$,
hence $\vdIM \fint(\G) \land \Box \psi \imp \Box \phi = \fint(\G, \nest{\bl{\psi}}, \wh{\Box \phi})$.
If $\D = \emptyset$, then from $\vdIM \fint(\wh{\phi}) = \phi$ we get $\vdIM \top \imp \phi$, then 
by $\monbox$, $\vdIM \fint(\G) \land \Box \top \imp \Box \phi = \fint(\G, \enest, \wh{\Box \phi})$.

\medskip
\noindent
\textit{Case for rule $\bdiam$.} \; We have $\D = \wh{\psi}$ or $\D = \bl{\psi}$ or $\D = \emptyset$.
If $\D = \wh{\psi}$, then from $\vdIM \fint(\bl{\phi}, \wh{\psi}) = \phi \imp \psi$,
by $\mondiam$ we obtain $\vdIM \fint(\G) \land \diam \phi \imp \diam \psi = \fint(\G, \bl{\diam \phi}, \nest{\wh{\psi}})$.
If $\D = \bl{\psi}$, then from $\vdIM \fint(\bl{\phi}, \bl{\psi}) = \phi \land \psi \imp \bot$,
by the rule \str\ we obtain $\vdIM \diam \phi \land \Box \psi \imp \bot$,
hence $\vdIM \fint(\Lambda) \land \diam \phi \land \Box \psi \imp \fint(\Pi) = \fint(\Lambda, \bl{\diam \phi}, \nest{\bl{\psi}}, \Pi)$.
If $\D = \emptyset$, then from $\vdIM \fint(\bl{\phi}) = \phi \imp \bot$
we get $\vdIM \phi \land \top \imp \bot$,
then by \str,
$\vdIM \diam \phi \land \Box \top \imp \bot$,
hence 
$\vdIM \fint(\Lambda) \land \diam \phi \land \Box \top \imp \fint(\Pi) = \fint(\Lambda, \bl{\diam \phi}, \enest, \Pi)$.
\end{proof}

  A routine induction on the height of the derivation
  now entails the following soundness theorem:

\begin{theorem}[Soundness]
If $\vd_{\CIM} \G$ then $\vdIM \fint(\G)$.
\end{theorem}

%==============================================================================%
\subsection{Structural properties and cut elimination}

We prove standard structural properties of the calculus $\CIM$, most importantly
admissibility of cut.
As a consequence, we obtain that $\CIM$ is complete with respect to $\IM$ and that
$\IM$ is decidable.
As usual, a rule is called \emph{admissible} in $\CIM$ if derivability of
the premises implies derivability of the conclusion.
A one-premiss rule is \emph{height-preserving admissible} (hp-admissible)
if derivability of the premiss via a derivation of height $h$
implies derivability of the conclusion with a derivation of height $\leq h$.
Finally, a rule $\Rule = (\G_1, ..., \G_n / \G)$ is \emph{height-preserving invertible} (hp-invertible)
if the rule $\Rulein = (\G/ \G_k)$ is hp-admissible for all $1 \leq k \leq n$.

\begin{proposition}\label{prop:gen ax}
The sequent $\Lambda, \bl{\phi}, \wh{\phi}$ is derivable in $\CIM$ for all $\Lambda, \phi$.
\end{proposition}
\begin{proof}
  By induction on the complexity $c$ of $\phi$.
  If $c = 0$, then $\Lambda, \bl{\phi}, \wh{\phi}$ is derived via $\id$ or $\bbot$.
  If $\phi = \Box \psi$, then
  by the induction hypothesis, ($\bl{\psi}, \wh{\psi}$) is derivable in $\CIM$.
  By applying $\wbox$ to this sequent we can then obtain
  ($\Lambda, \nest{\bl{\psi}}, \wh{\Box \psi}$),
  and by $\bbox$, ($\Lambda, \bl{\Box \psi}, \wh{\Box \psi}$),
  which means that
  $\Lambda, \bl{\Box \psi}, \wh{\Box \psi}$ is derivable as well.
  All other cases are similar.
\end{proof}

\begin{proposition}\label{prop:hp-inv-rules}
The rules $\bland$, $\wland$, $\blor$, $\wlor$, $\wimp$, $\bbox$, $\wdiam$ are height-preserving invertible.
The rule $\bimp$ is height-preserving invertible with respect to the right premiss.
\end{proposition}
\begin{proof}
For all rules, the proof is by induction on the height $h$ of the derivation of the conclusion.
We show as an example the case of $\bbox$.
Suppose ($\G, \bl{\Box \phi}$) is derivable with height $h$.
If $h = 0$, then ($\G, \nest{\bl{\phi}}$) is also derivable with height $0$.
If $h > 0$, we have two cases.
(1) $\bl{\Box \phi}$ is not principal in the last rule application $\Rule$ in the derivation of ($\G, \bl{\Box \phi}$).
If $\Rule = \wbox$ or $\Rule = \bdiam$, we apply $\Rule$ to its premiss introducing $\nest{\bl{\phi}}$ instead of $\bl{\Box \phi}$.
Otherwise we apply the induction hypothesis to the premiss(es) of $\Rule$, then we apply $\Rule$, obtaining ($\G, \nest{\bl{\phi}}$).
(2) $\bl{\Box \phi}$ is principal in $\Rule$. Then ($\G, \bl{\Box \phi}$) is obtained via $\bbox$ from ($\G, \nest{\bl{\phi}}$),
which means that ($\G, \nest{\bl{\phi}}$) is derivable with height $h' < h$.
\end{proof}

\begin{proposition}\label{prop:adm struct}
The following rules are height-preserving admissible in $\CIM$,
where $x$ is either a formula or a block:
\begin{equation*}
\ax{$\G$}
\llab{$\bwk$}
\uinf{$\G, \bl{x}$}
\disp
\qquad
\ax{$\Lambda$}
\llab{$\wwk$}
\uinf{$\Lambda, \wh{x}$}
\disp
\qquad
\ax{$\G, \bl{x}, \bl{x}$}
\llab{$\bctr$}
\uinf{$\G, \bl{x}$}
\disp
\qquad
\ax{$\G, \enest, \B$}
\llab{$\nelim$}
\uinf{$\G, \B$}
\disp
\end{equation*}
\end{proposition}

\begin{proof}
For all rules, we use induction on the height $h$ of the derivation of the premiss.

\medskip
\noindent
\textit{Cases for $\bwk$ and $\wwk$.} \;
If $h = 0$, then ($\G, \bl{x}$) (resp.~($\Lambda, \wh{x}$)) is derivable with height 0.
If $h > 0$, consider the last rule application $\Rule$ in the derivation of $\G$ ($\Lambda$).
If $\Rule = \wbox$ or $\Rule = \bdiam$, we apply $\Rule$ to its premiss introducing also $\bl{x}$ ($\wh{x}$).
Otherwise we apply $\bwk$ ($\wwk$) to the premisses of $\Rule$ (which is hp-admissible by the induction hypothesis), then we apply $\Rule$.

\medskip
\noindent
\textit{Case for $\bctr$.} \;
If $h = 0$, then ($\G, \bl{x}$) is derivable with height 0.
If $h > 0$, we consider the last rule application $\Rule$ in the derivation of ($\G, \bl{x}, \bl{x}$).
We distinguish two cases.
(1) $\bl{x}$ is not principal in $\Rule$.
If $\Rule = \wbox$ or $\Rule = \bdiam$, we apply $\Rule$ to its premiss introducing one occurrence of $\bl{x}$ only.
Otherwise we apply $\bctr$ to the premiss(es) of $\Rule$ (hp-admissible by the induction hypothesis), then we apply $\Rule$.
(2) $\bl{x}$ is principal in $\Rule$.
We consider the following three examples.

(1.~$\bl{x} = \bl{\phi \imp \psi}$)
We have ($\G, \bl{\phi \imp \psi}, \bl{\phi \imp \psi}$) obtained via $\bimp$ from
($\pr{\G}, \bl{\phi \imp \psi}, \bl{\phi \imp \psi}, \wh{\psi}$)
and ($\G, \bl{\psi}, \bl{\phi \imp \psi}$).
On the first premiss, we apply $\bctr$ (hp-admissible by the induction hypothesis)
and obtain
$\G_1 = (\pr{\G}, \bl{\phi \imp \psi}, \wh{\psi})$.
On the second premiss,
we first consider the hp-invertibility of $\bimp$ with respect to this premiss (Proposition~\ref{prop:hp-inv-rules}),
obtaining ($\G, \bl{\psi}, \bl{\psi}$),
then we apply $\bctr$ (hp-admissible by the induction hypothesis)
and obtain $\G_2 = (\G, \bl{\psi}$).
With an application of $\bimp$ on $\G_1$ and $\G_2$
we then obtain ($\G, \bl{\phi \imp \psi}$).

(2.~$\bl{x} = \bl{\Box \phi}$)
We have ($\G, \bl{\Box \phi}, \bl{\Box \phi}$) obtained via $\bbox$ from
($\G, \nest{\bl{\phi}}, \bl{\Box \phi}$).
Considering the hp-invertibility of $\bbox$ (Proposition~\ref{prop:hp-inv-rules}), 
from the premiss we obtain  ($\G, \nest{\bl{\phi}}, \nest{\bl{\phi}}$),
then by $\bctr$ (hp-admissible by the induction hypothesis)
we get ($\G, \nest{\bl{\phi}}$),
finally by $\bbox$ we have ($\G, \bl{\Box \phi}$). 

(3.~$\bl{x} = \nest{\bl{\phi}}$) 
There are two cases to consider.
(3.1)
We have ($\G', \nest{\bl{\phi}}, \nest{\bl{\phi}}, \wh{\Box \psi}$)
obtained via $\wbox$ from ($\bl{\phi}, \wh{\psi}$).
We can apply $\wbox$ on this premiss without introducing the second occurrence of $\nest{\bl{\phi}}$,
obtaining ($\G', \nest{\bl{\phi}}, \wh{\Box \psi}$).
(3.2)
We have ($\G', \nest{\bl{\phi}}, \nest{\bl{\phi}}, \bl{\diam \psi}$) 
obtained via $\bdiam$ from ($\bl{\phi}, \bl{\psi}$).
Again, we can apply $\bdiam$ on this premiss and obtain
($\G', \nest{\bl{\phi}}, \bl{\diam \psi}$).

\medskip
\noindent
\textit{Case for $\nelim$.} \;
If $h = 0$, then $(\G, \B)$ is derivable with height 0.
If $h > 0$, we consider the last rule application $\Rule$ in the derivation of $(\G, \enest, \B)$.
We distinguish three cases.

(1) If $\enest$ is principal in $\Rule$, then $\Rule = \wbox$ or $\Rule = \bdiam$.
If $\Rule = \wbox$, we have 
($\G', \enest, \B, \wh{\Box \psi}$)
obtained from ($\wh{\psi}$),
where $(\G', \wh{\Box \psi}) = \G$ and $\B = \nest{\bl{\phi}}$ or $\B = \enest$.
If $\B = \nest{\bl{\phi}}$,
then from ($\wh{\psi}$), by $\bwk$ we obtain $(\bl{\phi}, \wh{\psi})$, and by $\wbox$, 
($\G', \nest{\bl{\phi}}, \wh{\Box \psi}$).
If $\B = \enest$, 
then from ($\wh{\psi}$), 
by $\wbox$ we immediately obtain ($\G', \enest, \wh{\Box \psi}$).
In either case, the consequence ($\G, \B$) of $\nelim$ is derivable with height $h$.
The case $\Rule = \bdiam$ is similar.

(2) If $\B$ is principal in $\Rule$, then again $\Rule = \wbox$ or $\Rule = \bdiam$.
We obtain $\G, \B$ by applying the same rule $\Rule$ but without introducing $\enest$. 

(3) If neither $\enest$ nor $\B$ is principal in $\Rule$,
then in case $\Rule = \wbox$ or $\Rule = \bdiam$, we apply $\Rule$ without introducing $\enest$.
Otherwise we reverse the order of the applications of $\Rule$ and $\nelim$,
applying $\nelim$ on the premiss(es) of $\Rule$ 
(with $\nelim$ hp-admissible by the induction hypothesis).
For instance, suppose that ($\G, \enest, \B$) is obtained via $\bimp$.
Then $\G = (\G', \bl{\phi \imp \psi})$ and $\B = \nest{\wh{\chi}}$,
and the premisses of $\bimp$ have the form
($\G', \bl{\phi \imp \psi}, \enest, \enest, \wh{\phi}$)
and 
($\G', \bl{\psi}, \enest, \nest{\wh{\chi}}$).
By applying $\nelim$ to each premiss we then obtain
($\G', \bl{\phi \imp \psi}, \enest, \wh{\phi}$)
and 
($\G', \bl{\psi}, \nest{\wh{\chi}}$),
thus by $\bimp$ we get
$(\G', \bl{\phi \imp \psi}, \nest{\wh{\chi}}) = (\G, \B)$.
\end{proof}

\begin{theorem}[Cut admissibility]\label{th:cut}
The following rules are admissible in $\CIM$,
where $|\Lambda| \leq 1$ in $\bsub$,
and $|\D| \leq 1$ in $\wsub$:
\begin{equation*}
\ax{$\Lambda, \wh{\phi}$}
\ax{$\bl{\phi}, \D$}
\llab{$\cut$}
\binf{$\Lambda, \D$}
\disp 
\qquad
\ax{$\Lambda, \wh{\phi}$}
\ax{$\nest{\bl{\phi}}, \D$}
\llab{$\bsub$}
\binf{$\nest{\Lambda}, \D$}
\disp
\qquad
\ax{$\Lambda, \nest{\wh{\phi}}$}
\ax{$\bl{\phi}, \D$}
\llab{$\wsub$}
\binf{$\Lambda, \nest{\D}$}
\disp
\end{equation*}
\end{theorem}

\begin{proof} %[Proof sketch]
For an application of $\cut$, we call \emph{cut formula} the formula $\phi$ deleted by that application,
while for $\bsub$, $\wsub$ we call \emph{subs-formula} the formula $\phi$ inside the block replaced by $\Lambda$, resp.~$\D$.
Moreover, we call \emph{cut height} the sum of the heights of the derivations of the premisses of $\cut$.
The theorem is a consequence of the following four claims, where
$Cut(c, h)$ means that all applications of $\cut$ of height $h$ with cut formula of complexity $c$ are admissible, and
$Sub(c)$ means that all applications of $\bsub$ and $\wsub$ with subs-formula of complexity $c$ are admissible:

\begin{center}
\begin{tabular}{l}
\textbf{(A)} $\forall c. Cut(c,0)$.
\quad
\textbf{(B)} $\forall h. Cut(0,h)$.
\quad
\textbf{(C)} $\forall c. (\forall h. Cut(c, h) \to Sub(c))$.
\\
\textbf{(D)} $\forall c. \forall h. (\forall c' < c. (Sub(c') \land \forall h'. Cut(c', h')) \land \forall h'' < h. Cut(c, h'') \to Cut(c, h))$.
\end{tabular}
\end{center}

\medskip
\noindent
\textbf{(A)} \;
If $h = 0$, then the premisses of $\cut$ are instances of $\id$ or $\bbot$.
In either case the consequence of $\cut$ is an instance of $\id$ or $\bbot$.

\medskip
\noindent
\textbf{(B)} \; 
If $c = 0$, then the cut formula if $\bot$ or a propositional variable $p$.
The proof proceeds by complete induction on $h$.
For the base case $h = 0$, then the claim follows from \textbf{(A)}.
For the inductive step, we distinguish three cases:
(i) The cut formula is not principal in the last rule $\Rule$ applied in the derivation of the left premiss of $\cut$.
We consider the following three examples.

\medskip
($\Rule = \bimp$) 
%\vspace{-1.7em}
%\begin{center}
\qquad
\begin{small}
\ax{$\Lambda', \bl{\psi \imp \chi}, \wh{\psi}$}
\ax{$\Lambda', \bl{\chi}, \wh{p}$}
\llab{$\bimp$}
\binf{$\Lambda', \bl{\psi \imp \chi}, \wh{p}$}
\ax{$\bl{p}, \D$}
\llab{$\cut$}
\binf{$\Lambda', \bl{\psi \imp \chi}, \D$}
\disp
\quad
(where $\Lambda' = \pr{(\Lambda')}$)
\end{small}

\smallskip
\noindent
The derivation is transformed as follows.
From ($\Lambda', \bl{\psi \imp \chi}, \wh{\psi}$),
by some applications of $\bwk$ we introduce $\pr{\D}$,
obtaining $\G_1 = (\Lambda', \pr{\D}, \bl{\psi \imp \chi}, \wh{\psi})$.
Moreover, from ($\Lambda', \bl{\chi}, \wh{p}$) and ($\bl{p}, \D$),
by $\cut$ we obtain $\G_2 = (\Lambda', \bl{\chi}, \D)$.
Finally, from $\G_1$ and $\G_2$, 
since $\Lambda', \pr{\D} = \pr{(\Lambda', \D)}$,
by $\bimp$ we obtain ($\Lambda', \bl{\psi \imp \chi}, \D$).

\medskip
($\Rule = \bbox$) 
\qquad
%\vspace{-1.7em}
%\begin{center}
\begin{small}
\ax{$\Lambda', \nest{\bl{\psi}}, \wh{p}$}
\llab{$\bbox$}
\uinf{$\Lambda', \bl{\Box \psi}, \wh{p}$}
\ax{$\bl{p}, \D$}
\llab{$\cut$}
\binf{$\Lambda', \bl{\Box \psi}, \D$}
\disp
\end{small}
%\end{center}

\smallskip\noindent
The derivation is transformed as follows.
From ($\Lambda',  \nest{\bl{\psi}}, \wh{p}$) and ($\bl{p}, \D$),
by $\cut$ we obtain ($\Lambda', \nest{\bl{\psi}}, \D$),
then by $\bbox$ we obtain ($\Lambda', \bl{\Box \psi}, \D$).

\medskip
($\Rule = \bdiam$)
\qquad
%\vspace{-1.7em}
%\begin{center}
\begin{small}
\ax{$\bl{\psi}, \Sigma$}
\llab{$\bdiam$}
\uinf{$\Lambda', \bl{\diam \psi}, \nest{\Sigma}, \wh{p}$}
\ax{$\bl{p}, \D$}
\llab{$\cut$}
\binf{$\Lambda', \bl{\diam \psi}, \nest{\Sigma}, \D$}
\disp
\end{small}
%\end{center}

\smallskip
\noindent
The derivation is transformed as follows.
From ($\bl{\psi}, \Sigma$), by $\bdiam$ we obtain ($\Lambda', \bl{\diam \psi}, \nest{\Sigma}, \D$).

\bigskip
\noindent
(ii) The cut formula is not principal in the last rule $\Rule$ applied in the derivation of the right premiss of $\cut$.
We consider the following examples.

\medskip
($\Rule = \bimp$) 
\qquad
%\vspace{-1.7em}
%\begin{center}
\begin{footnotesize}
\ax{$\Lambda, \wh{p}$}
\ax{$\bl{p}, \pr{(\D')}, \bl{\psi \imp \chi}, \wh{\psi}$}
\ax{$\bl{p}, \D', \bl{\chi}$}
\rlab{$\bimp$}
\binf{$\bl{p}, \D', \bl{\psi\imp \chi}$}
\rlab{$\cut$}
\binf{$\Lambda, \D', \bl{\psi\imp \chi}$}
\disp
\end{footnotesize}
%\end{center}

\smallskip\noindent
The derivation is transformed as follows.
From ($\Lambda, \wh{p}$) and ($\bl{p}, \pr{(\D')}, \bl{\psi \imp \chi}, \wh{\psi}$),
by $\cut$ we obtain $\G_1 = (\Lambda, \pr{(\D')}, \bl{\psi \imp \chi}, \wh{\psi})$.
Moreover, from 
($\Lambda, \wh{p}$) and ($\bl{p}, \D', \bl{\chi}$),
by $\cut$ we obtain $\G_2 = (\Lambda, \D', \bl{\chi})$.
Finally, from $\G_1$ and $\G_2$, by $\bimp$
we obtain ($\Lambda, \D', \bl{\psi \imp \chi}$)

\medskip
($\Rule = \wdiam$) 
\qquad
%\vspace{-1.7em}
%\begin{center}
\begin{small}
\ax{$\Lambda, \wh{p}$}
\ax{$\bl{p}, \D', \nest{\wh{\psi}}$}
\rlab{$\wdiam$}
\uinf{$\bl{p}, \D', \wh{\diam \psi}$}
\rlab{$\cut$}
\binf{$\Lambda, \D', \wh{\diam \psi}$}
\disp
\end{small}
%\end{center}

\smallskip\noindent
The derivation is transformed as follows.
From ($\Lambda, \wh{p}$) and ($\bl{p}, \D', \nest{\wh{\psi}}$),
by $\cut$ we obtain ($\Lambda, \D', \nest{\wh{\psi}}$),
then by $\wdiam$ we obtain ($\Lambda, \D', \wh{\diam \psi}$).

\bigskip
\noindent
(iii) The cut formula is principal on both sides.
Then 
both premisses of $\cut$ are obtained via $\id$, 
hence $h = 0$ and we are back to the case \textbf{(A)}.

\medskip
\noindent
\textbf{(C)} \;
Assume $\forall h. Cut(c, h)$. We need to prove that all applications of $\bsub$ and $\wsub$ 
where $\phi$ has complexity $c$ are admissible.
The proof is by induction on the height $j$ of the derivation $\Der$ of ($\nest{\bl{\phi}}, \D$), for $\bsub$,
and of ($\Lambda, \nest{\wh{\phi}}$), for $\wsub$.
If $j = 0$ or $\nest{\bl{\phi}}$ (resp. $\nest{\wh{\phi}}$) is not principal in the last rule application $\Rule$ of $\Der$, we proceed %as before. \nb{T: specify: case (B)?}
as in \textbf{(A)} or \textbf{(B)}.
If $j > 0$ and $\nest{\bl{\phi}}$ (resp. $\nest{\wh{\phi}}$) is principal in $\Rule$,
then $\Rule = \wbox$ or $\Rule = \bdiam$.
%The case $\Rule = \wbox$ was shown in the main text (proof sketch of Th~\ref{th:cut}).

Suppose $\Rule = \wbox$.
Then, only the premiss ($\nest{\bl{\phi}}, \D$) of $\bsub$ is derivable via $\Rule = \wbox$ 
(the premiss ($\Lambda, \nest{\wh{\phi}}$) of $\wsub$ is not).
In this case,
($\nest{\bl{\phi}}, \D$) has the form ($\nest{\bl{\phi}}, \D', \wh{\Box \chi}$),
derived via $\wbox$ from ($\bl{\phi}, \wh{\chi}$).
We can then obtain the consequence ($\nest{\Lambda}, \D', \wh{\Box \chi}$) of $\bsub$
as follows.
If $\Lambda = \emptyset$, then
the other premiss of $\bsub$ is ($\wh{\phi}$). 
From ($\wh{\phi}$) and ($\bl{\phi}, \wh{\chi}$),
by $\cut$ we obtain ($\wh{\chi}$), and by $\wbox$ we obtain ($\enest, \D', \wh{\Box \chi}$).
If instead $\Lambda = \bl{\psi}$, then
the other premiss of $\bsub$ is ($\bl{\psi}, \wh{\phi}$). 
From ($\bl{\psi}, \wh{\phi}$) and ($\bl{\phi}, \wh{\chi}$), 
by $\cut$ we obtain ($\bl{\psi}, \wh{\chi}$), and by $\wbox$ we obtain ($\nest{ \bl{\psi}}, \D', \wh{\Box \chi}$).
In both cases, the application of $\cut$ is admissible by the induction hypothesis.

Now suppose $\Rule = \bdiam$.
Let us consider first
the premiss ($\nest{\bl{\phi}}, \D$) of $\bsub$,
that in this case has the form ($\nest{\bl{\phi}}, \D', \bl{\diam \chi}$),
derived via $\bdiam$ from ($\bl{\phi}, \bl{\chi}$).
We obtain the consequence ($\nest{\Lambda}, \D', \bl{\diam \chi}$) of $\bsub$ as follows.
If $\Lambda = \emptyset$, then from ($\wh{\phi}$) and ($\bl{\phi}, \bl{\chi}$), by $\cut$
we obtain ($\bl{\chi}$), then by $\bdiam$ we obtain ($\enest, \D', \bl{\diam \chi}$).
If $\Lambda = \bl{\psi}$, then from ($\bl{\psi}, \wh{\phi}$) and ($\bl{\phi}, \bl{\chi}$),
by $\cut$ we obtain ($\bl{\psi}, \bl{\chi}$), then by $\bdiam$
we obtain ($\nest{ \bl{\psi}}, \D', \bl{\diam \chi}$).
In both cases, the application of $\cut$ is admissible by the induction hypothesis.

Let us now consider
the premiss ($\Lambda, \nest{\wh{\phi}}$) of $\wsub$,
that in this case has the form ($\Lambda', \bl{\diam \chi}, \nest{\wh{\phi}}$),
derived via $\bdiam$ from ($\bl{\chi}, \wh{\phi}$).
We obtain  the consequence ($\Lambda', \bl{\diam \chi}, \nest{\D}$) of $\wsub$ as follows.
If $\D = \emptyset$, then from ($\bl{\chi}, \wh{\phi}$) and ($\bl{\phi}$), by $\cut$
we obtain ($\bl{\chi}$), then by $\bdiam$ we obtain ($\Lambda', \bl{\diam \chi}, \enest$).
If $\D = \wh{\psi}$, then from ($\bl{\chi}, \wh{\phi}$) and ($\bl{\phi}, \wh{\psi}$),
by $\cut$ we obtain ($\bl{\chi}, \wh{\psi}$), then by $\bdiam$
we obtain ($\Lambda', \bl{\diam \chi}, \nest{\wh{\psi}}$).
In both cases, the application of $\cut$ is admissible by the induction hypothesis.

\medskip
\noindent
\textbf{(D)} \;
Assume $\forall c' < c. (Sub(c') \land \forall h'. Cut(c', h')) \land \forall h'' < h. Cut(c, h'')$. 
We show that all applications of $\cut$ of height $h$ on a cut formula of complexity  $c$
can be replaced by different applications of $\cut$ of smaller height or on a cut formula of smaller complexity. 
We can assume $h,c > 0$ as the cases $h = 0$ and $c = 0$ have been already considered in \textbf{(A)} and \textbf{(B)}. 
We distinguish two cases:
(i) The cut formula is not principal in the last rule application in the derivation of at least one of  the two premisses of $\cut$. This case is analogous to (i) or (ii) in \textbf{(B)}.
(ii) The cut formula $\phi$ is principal in the last rule application in the derivation of both premisses.
We show the cases $\phi = \psi \imp \chi, \Box \psi, \diam \psi$.
For $\phi = \psi \imp \chi$, we consider two subcases, depending on whether the output component is a formula or a block.

\medskip
($\phi = \psi \imp \chi$, i) 
\qquad
\begin{small}
\ax{$\Lambda, \bl{\psi}, \wh{\chi}$}
\llab{$\wimp$}
\uinf{$\Lambda, \wh{\psi \imp \chi}$}
\ax{$\bl{\psi \imp \chi}, \D', \wh{\psi}$}
\ax{$\bl{\chi}, \D', \wh{\theta}$}
\rlab{$\bimp$ ($\D' = \pr{(\D')}$)}
\binf{$\bl{\psi \imp \chi}, \D', \wh{\theta}$}
\rlab{$\cut$}
\binf{$\Lambda, \D', \wh{\theta}$}
\disp
\end{small}
%\end{center}

\smallskip\noindent
The derivation is transformed as follows.
First, we cut $\psi \imp \chi$
from 
($\Lambda, \wh{\psi \imp \chi}$) and ($\bl{\psi \imp \chi}, \D' \wh{\psi}$),
obtaining ($\Lambda, \D', \wh{\psi}$).
From this and ($\Lambda, \bl{\psi}, \wh{\chi}$),
we cut $\psi$ and obtain ($\Lambda, \Lambda, \D', \wh{\chi}$).
From this and ($\bl{\chi}, \D', \wh{\theta}$) we cut $\chi$ and obtain
($\Lambda, \Lambda, \D', \D', \wh{\theta}$).
Finally, by some applications of $\bctr$
we obtain ($\Lambda, \D', \wh{\theta}$).

\medskip
($\phi = \psi \imp \chi$, ii) 
\qquad
\begin{small}
\ax{$\Lambda, \bl{\psi}, \wh{\chi}$}
\llab{$\wimp$}
\uinf{$\Lambda, \wh{\psi \imp \chi}$}
\ax{$\bl{\psi \imp \chi}, \D', \enest, \wh{\psi}$}
\ax{$\bl{\chi}, \D', \nest{\wh{\theta}}$}
\rlab{$\bimp$}
\binf{$\bl{\psi \imp \chi}, \D', \nest{\wh{\theta}}$}
\rlab{$\cut$}
\binf{$\Lambda, \D', \nest{\wh{\theta}}$}
\disp
\end{small}
%\end{center}

\smallskip\noindent
The derivation is transformed as follows.
First, we cut $\psi \imp \chi$
from 
($\Lambda, \wh{\psi \imp \chi}$) and ($\bl{\psi \imp \chi}, \D', \enest, \wh{\psi}$),
obtaining ($\Lambda, \D', \enest, \wh{\psi}$).
From this and ($\Lambda, \bl{\psi}, \wh{\chi}$),
we cut $\psi$ and obtain ($\Lambda, \Lambda, \D', \enest, \wh{\chi}$).
From this and ($\bl{\chi}, \D', \nest{\wh{\theta}}$) we cut $\chi$ and obtain
($\Lambda, \Lambda, \D', \D', \enest, \nest{\wh{\theta}}$).
Finally, by one application of $\nelim$ and some applications of $\bctr$,
we obtain ($\Lambda, \D', \nest{\wh{\theta}}$).

\medskip
($\phi = \Box \psi$) 
%\vspace{-1.7em}
%\begin{center}
\qquad
\begin{small}
\ax{$\Sigma, \wh{\psi}$}
\llab{$\wbox$}
\uinf{$\Lambda', \nest{\Sigma}, \wh{\Box \psi}$}
\ax{$\nest{\bl{\psi}}, \D$}
\rlab{$\bbox$}
\uinf{$\bl{\Box \psi}, \D$}
\llab{$\cut$}
\binf{$\Lambda', \nest{\Sigma}, \D$}
\disp
\end{small}
%\end{center}

\smallskip\noindent
The derivation is transformed as follows.
From ($\Sigma, \wh{\psi}$) and ($\nest{\bl{\psi}}, \D$),
by $\bsub$ we obtain ($\nest{\Sigma}, \D$),
then by $\swk$ we obtain ($\Lambda', \nest{\Sigma}, \D$).

\medskip
($\phi = \diam \psi$) 
%\vspace{-1.7em}
%\begin{center}
\qquad
\begin{small}
\ax{$\Lambda, \nest{\wh{\psi}}$}
\llab{$\wdiam$}
\uinf{$\Lambda, \wh{\diam \psi}$}
\ax{$\bl{\psi}, \Sigma$}
\rlab{$\bdiam$}
\uinf{$\bl{\diam \psi}, \D', \nest{\Sigma}$}
\llab{$\cut$}
\binf{$\Lambda, \D', \nest{\Sigma}$}
\disp
\end{small}
%\end{center}

\smallskip\noindent
The derivation is transformed as follows.
From ($\Lambda, \nest{\wh{\psi}}$) and ($\bl{\psi}, \Sigma$),
by $\wsub$ we obtain ($\Lambda, \nest{\Sigma}$),
then by $\swk$ we obtain ($\Lambda, \D', \nest{\Sigma}$).
\end{proof}

\begin{theorem}[Completeness]
If $\vdIM \phi$, then $\vd_{\CIM} \wh{\phi}$.
\end{theorem}
\begin{proof}
  This follows from induction on the height of the derivation of $\phi$ in $\IM$, 
using the fact that the axioms of $\IM$ are derivable in $\CIM$,
and the rules $\monbox$, $\mondiam$ and modus ponens are simulated in $\CIM$ via
$\cut$.
For instance,
from the premiss ($\wh{\phi \imp \psi}$)
and the derivable sequent ($\bl{\phi \imp \psi}, \bl{\phi}, \wh{\psi}$),
cutting on $\phi \imp \psi$ gives ($\bl{\phi}, \wh{\psi}$).
Applying $\bdiam$ and $\wdiam$ yields ($\bl{\diam \phi}, \wh{\diam \psi}$),
so that $\wimp$ entails ($\wh{\diam \phi \imp \diam \psi}$).
\end{proof}

As is often the case, one can infer important properties of a logic based on its analytic cut-free calculus. As a first application of the calculus, we prove here that $\IM$ is decidable.

\begin{theorem}[Decidability]\label{th:decidability}
  For any $\phi \in \mc{L}$, it is decidable whether or not $\IM \vdash \phi$.
\end{theorem}

\begin{proof}
By adapting the decidability argument for $\mathsf{G3ip}$ in \cite{TroSch00}, 
we can prove that for any sequent $\G$, it is decidable whether or not $\G$ is derivable in $\CIM$.
 To this end, consider that
 (1) for all rules $\Rule$ of $\CIM$, 
 the premiss(es) of $\Rule$ only contain subfomulas
 of formulas in the conclusion;
 (2) with the exception of $\bimp$, the backward application
 of all rules of $\CIM$ reduces the sequent complexity 
 according to the following measure
 $C(\G) = \sum_{x \in \G} c'(x)$, where
 for $\ast \in \{\bullet, \circ\}$ and
 $\star \in \{ \wedge, \vee, \imp \}$,
 $c'(\phi^\ast) = c(\phi)$
 if $\phi = p, \bot, \psi \star \chi$;
 $c'(\Box\phi^\ast) = c'(\diam\phi^*) = c'(\phi^*) + 2$;
 $c'(\enest) = 1$;
 $c'(\nest{\phi^*}) = c'(\phi^*) + 1$.
 Moreover, for a sequent $\G$, let us by denote $Set(\G)$
 the support of $\G$, that is,
 the set of formulas and blocks in $\G$ disregarding multiplicities.
 We can then consider bottom-up proof search in $\CIM$
 respecting the following \emph{loop checking} condition LC:
 the backward application of $\bimp$ is not allowed 
 if a premiss $\G$ of $\bimp$ is such that
 $Set(\G) = Set(\D)$
 for a sequent $\D$ already occurring at a lower depth in the same proof-search branch.
Because of the subformula property, only a finite number of applications of $\bimp$ are possible respecting LC.
Therefore, bottom-up proof search eventually terminates on every $\G$.
As usual, the completeness of such a proof-search procedure is ensured by the admissibility of contraction in $\CIM$.
\end{proof}

%==============================================================================%
\subsection{A calculus for $\WM$}\label{sec:CWM}

  Interestingly, $\CIM$ can be turned into a calculus
  for $\WM$ by changing from
  \emph{output pruning} $\Gamma^{\downarrow}$ in $\bimp$ to \emph{output thinning}%
    \footnote{``A more drastic form of pruning, a thinning out cut, is the
                removal of an entire shoot, limb, or branch at its point of
                origin'' (https://en.wikipedia.org/wiki/Pruning).}
  $\thi{\G}$,
  which is defined by not only removing output formulas, but also output blocks.
  For instance, $\thi{(\bl{\phi}, \nest{\bl{\psi}}, \wh{\chi})} = \thi{(\bl{\phi}, \nest{\bl{\psi}}, \nest{\wh{\chi}})} = \bl{\phi}, \nest{\bl{\psi}}$.
  The calculus $\CWM$ is defined by the rules in Figure~\ref{fig:CIM}, where context pruning in $\bimp$ is replaced with context thinning.

$\CWM$ can be shown equivalent to $\WM$ by adjusting the structural analysis of $\CIM$
in all cases where pruning comes into play.
Note that empty blocks never occur in derivations of formulas, and could therefore be disallowed.
If empty blocks are preserved, we observe that the rule $\nelim$ is no longer admissible
(one can verify that the soundness of $\nelim$ depends on the presence of $\axIdiam$) 
but it is not needed for cut elimination.
Alternatively, completeness of $\CWM$ can be proved by translating derivations in the sequent calculus $\mathsf{S.WM}$ \cite{Dal22} into derivations in $\CWM$,
with applications of the single $\Box$-rule (respectively $\diam$-rule) of $\mathsf{S.WM}$ corresponding to an application of $\wbox$ ($\bdiam$)
immediately followed by an application of $\bbox$ ($\wdiam$).

It was shown in \cite{Gro25-im} that $\IM$ and $\WM$ share the same $\diam$-free fragment, but have different $\Box$-free fragments.
We can re-prove the same results using the calculi $\CIM$ and $\CWM$.
First, note that for every $\diam$-free formula $\phi$ derivable in $\IM$,
a bottom-up derivation $\Der$ in $\CIM$ of $\wh{\phi}$ never introduces output blocks, 
which means that context pruning and thinning are indistinguishable in $\Der$,
and $\Der$ can be therefore seen as a derivation in $\CWM$.
On the other hand, 
using the bottom-up proof-search strategy for $\CIM$ given in the proof of Theorem~\ref{th:decidability}, 
which is equally applicable to $\CWM$,
one can verify that the formula $((\diam\bot \imp \bot) \imp \bot) \imp\diam\bot$ is not derivable in $\WM$,
as the proof-search deadlocks on a sequent of the form 
($\bl{(\diam\bot \imp \bot) \imp \bot}, \bl{\diam\bot}, \wh{\diam\bot\imp\bot}$).
Example~\ref{exm:der} shows that this formula is derivable in $\IM$.

%%%%%%%%%%%%%%%%%%%%%%%%%%%%%%%%%%%%%%%%%%%%%%%%%%%%%%%%%%%%%%%%%%%%%%%%%%%%%%%%
\section{Extensions}

We now present calculi for the axiomatic extensions of $\IM$ with 
axioms from $\{ \axN, \axP, \axT, \axD, \axK \}$. 
The design of these follows \emph{Dosen's principle} \cite{Wan94}:
calculi for extensions are defined by modularly adding rules dealing with structural components, 
while the rules for logical connectives remain untouched.
As for $\CIM$, these calculi can be seen as 
single-succedent restrictions of corresponding classical calculi from~\cite{DalLelOliPim21}.

\subsection{Calculi for extensions without $\axKbd$}

Calculi for extensions of $\IM$ are characterised by the following rules: 

\begin{center}
\ax{$\G, \nest{ \ }$}
\llab{$\rulen$}
\uinf{$\G$}
\disp
\quad
\ax{$\Dp$}
\llab{$\rulep$}
\uinf{$\G, \nest{\Dp}$}
\disp
\quad
\ax{$\Dp, \Sp$}
\llab{$\ruled$}
\uinf{$\G, \nest{\Dp}, \nest{\Sp}$}
\disp
\quad
\ax{$\G, \pr{\nest{\Dp}}, \Dp$}
\llab{$\rulet$}
\uinf{$\G, \nest{\Dp}$}
\disp
\quad
($|\Dp| = |\Sp| = 1$)
\end{center}

\begin{definition}
For any logic $\IMstar$, 
with $\Ax \subseteq \{ \axN, \axP, \axD, \axTbd \}$,
the calculus $\CIMstar$ for $\IMstar$ is defined as the extension of $\CIM$
with the corresponding rules 
among $\rulen$, $\rulep$, $\ruled$, $\rulet$ above.
If $\axD \in \Ax$, the corresponding calculus additionally contains the rule $\rulep$.
\end{definition}

\begin{example}
  Sequents corresponding to some of the axioms can be derived in $\CIMstar$ as follows:
  \begin{center}
\ax{$\wh{\top}$}
\rlab{$\wbox$}
\uinf{$\nest{ \ }, \wh{\Box\top}$}
\rlab{$\rulen$}
\uinf{$\wh{\Box\top}$}
\disp
\quad
\ax{$\wh{\top}$}
\rlab{$\rulep$}
\uinf{$\nest{\wh{\top}}$}
\rlab{$\wdiam$}
\uinf{$\wh{\diam\top}$}
\disp
\quad
\ax{$\bl{\phi}, \wh{\phi}$}
\rlab{$\ruled$}
\uinf{$\nest{\bl{\phi}}, \nest{\wh{\phi}}$}
\rlab{$\wdiam$}
\uinf{$\nest{\bl{\phi}}, \wh{\diam \phi}$}
\rlab{$\bbox$}
\uinf{$\bl{\Box \phi}, \wh{\diam \phi}$}
\rlab{$\wimp$}
\uinf{$\wh{\Box \phi \imp \diam \phi}$}
\disp
\quad
\ax{$\nest{\bl{\phi}}, \bl{\phi}, \wh{\phi}$}
\rlab{$\rulet$}
\uinf{$\nest{\bl{\phi}}, \wh{\phi}$}
\rlab{$\bbox$}
\uinf{$\bl{\Box \phi}, \wh{\phi}$}
\rlab{$\wimp$}
\uinf{$\wh{\Box \phi \imp \phi}$}
\disp
\quad
\ax{$\bl{\phi}, \enest, \wh{\phi}$}
\rlab{$\rulet$}
\uinf{$\bl{\phi}, \nest{\wh{\phi}}$}
\rlab{$\wdiam$}
\uinf{$\bl{\phi}, \wh{\diam \phi}$}
\rlab{$\wimp$}
\uinf{$\wh{\phi \imp \diam \phi}$}
\disp
  \end{center}
\end{example}

\begin{remark} \
  \begin{enumerate}
    \item The rule $\rulep$ is added to the calculus in case $\axD \in \Ax$ 
          to ensure the admissibility of contraction.
          This addition is sound since $\vd_{\IM \oplus \axD} \axP$.
    \item Adding the rule $\rulen$ to either $\CIM$ or $\CWM$ generates
          the same calculus.
          Indeed, for every output block deleted via output thinning 
          in a backward application of $\bimp$,
          a corresponding empty block can be reintroduced using $\rulen$.
          This matches the axiomatic equivalence between $\IM \oplus \axN$ and $\WM \oplus \axN$.
  \end{enumerate}
\end{remark}

We can prove their soundness and completeness with respect to the corresponding logics $\IMstar$
by extending the proofs for $\CIM$ from Section~\ref{sec:CIM} 
with the additional cases determined by the new rules.

\begin{theorem}[Soundness]\label{thm:ext-sound}
If $\vd_{\CIMstar} \G$ then $\vd_{\IMstar} \fint(\G)$.
\end{theorem}
\begin{proof} %[Proof of Theorem~\ref{thm:ext-sound}]
For any pair $(\varax, \varrule) \in 
\{(\axN, \rulen), (\axP, \rulep), (\axD, \ruled), (\axT, \rulet)\}$,
we show that the rule $\varrule$ is sound with respect to $\IM \oplus \varax$.
We let $\G = \Lambda, \Pi$, and 
shorten $ ( \bigwedge_{x \in \Lambda} \fintf(x)) = \lambda$
and $\fintf(\Pi) = \pi$.

\medskip
\noindent
\textit{Case for rule $\rulen$.} \;
Suppose $\vd_{\IMN} \fint(\G, \enest) = \lambda \land \Box\top \imp \pi$.
Then by axiom $\axN$,
 $\vd_{\IMN} \lambda \imp \pi = \fint(\G)$.
 
\medskip
\noindent 
\textit{Case for rule $\rulep$.} \;
If $\Dp = \bl{\phi}$,
then from $\vd_{\IMP} \fint(\bl{\phi}) = \phi \imp \bot$,
by $\monbox$ we get $\vd_{\IMP} \Box\phi \imp \Box\bot$,
then by $\axPbox$,  $\vd_{\IMP} \Box\phi \imp \bot$,
hence 
$\vd_{\IMN} \lambda \land \Box\phi \imp \pi = \fint(\G, \nest{\bl{\phi}})$.
If $\Dp = \wh{\phi}$,
then $\G = \Lambda$,
moreover from $\vd_{\IMP} \fint(\wh{\phi}) = \phi$,
we get $\vd_{\IMP} \fint(\wh{\phi}) = \top\imp\phi$,
hence by $\mondiam$,
$\vd_{\IMP} \fint(\wh{\phi}) = \diam\top\imp\diam\phi$.
Thus, being $\lambda \imp \diam\top$ derivable from $\axPdiam$,
we get 
$\vd_{\IMP} \lambda \imp \diam\phi = \fint(\Lambda, \nest{\wh{\phi}})$.

\medskip
\noindent
\textit{Case for rule $\ruled$.} \;
If $\Dp = \bl{\phi}$ and $\Sp = \bl{\psi}$, 
then from $\vd_{\IMD} \fint(\bl{\phi}, \bl{\psi}) = \phi \land \psi \imp \bot$, by the rule $\RD$ (Lem\-ma~\ref{lemma:eq axioms}),
we get $\vd_{\IMD} \Box\phi \land \Box\psi \imp \bot$, hence
 $\vd_{\IMD} \lambda \land \Box\phi \land \Box\psi \imp \pi = \fint(\G, \nest{\bl{\phi}}, \nest{\bl{\psi}})$.
If (w.l.o.g.) $\Dp = \bl{\phi}$ and $\Sp = \wh{\psi}$, 
then $\G = \Lambda$,
moreover from $\vd_{\IMD} \fint(\bl{\phi}, \wh{\psi}) = \phi \imp \psi$,
by $\mondiam$ we get $\vd_{\IMD} \diam\phi \imp \diam\psi$, then by the axiom $\axD$,
$\vd_{\IMD} \Box\phi \imp \diam\psi$,
therefore 
$\vd_{\IMD} \lambda \land \Box\phi \imp \diam\psi = \fint(\Lambda, \nest{\bl{\phi}}, \nest{\wh{\psi}})$.

\medskip
\noindent
\textit{Case for rule $\rulet$.} \; 
If $\Dp = \bl{\phi}$,
then from
$\vd_{\IMT} \fint(\Lambda, \nest{\bl{\phi}}, \bl{\phi}, \Pi) = \lambda \land \Box\phi \land \phi \imp \pi$,
by axiom $\axTbox$ we get 
$\vd_{\IMT} \lambda \land \Box\phi \imp \pi = \fint(\Lambda, \nest{\bl{\phi}}, \Pi)$.
If $\Dp = \wh{\phi}$,
then $\G = \Lambda$, moreover from
$\vd_{\IMT} \fint(\Lambda, \enest, \wh{\phi}) = \lambda \land \Box\top \imp \phi$,
by axiom $\axTdiam$ we get 
$\vd_{\IMT} \lambda \land \Box\top \imp \diam\phi$,
thus  
$\vd_{\IMT} \lambda \imp (\Box\top \imp \diam\phi)$,
hence by $\axIdiam$, 
$\vd_{\IMT} \lambda \imp \diam\phi =  \fint(\Lambda, \nest{\wh{\phi}})$.
\end{proof}

\begin{proposition}\label{prop:ext-adm-rules}
The rules $\bwk$, $\wwk$, $\bctr$ and $\nelim$ are height-preserving admissible in $\CIMstar$.
\end{proposition}
\begin{proof}
By extending the proof of Proposition~\ref{prop:adm struct} with the additional cases determined by the new rules.
We show only a few significant cases.
For $\bctr$, if its premiss ($\G, \nest{\bl{\phi}}, \nest{\bl{\phi}}$) is derived via $\ruled$ from ($\bl{\phi}, \bl{\phi}$),
then by a hp-application of $\bctr$ to the premiss of $\ruled$ we obtain ($\bl{\phi}$), then by $\rulep$ we have ($\G, \nest{\bl{\phi}}$).
For $\nelim$, note that $\nest$ is never principal in (bottom-up) applications of the rules $\rulen, \rulep, \ruled, \rulet$.
If the premiss ($\G, \enest, \B$) of $\nelim$ is obtained via $\rulen$ from ($\G, \enest, \enest, \B$),
then by two hp-applications of $\nelim$ to this sequent we obtain ($\G, \B$).
If the premiss ($\Lambda, \enest, \nest{\wh{\phi}}$) of $\nelim$ is obtained via $\rulet$ from ($\Lambda, \enest, \enest, \wh{\phi}$),
then by a hp-application of $\nelim$ to this sequent we obtain ($\Lambda, \enest, \wh{\phi}$), and by $\rulet$
we obtain ($\Lambda, \nest{\wh{\phi}}$).
\end{proof}

\begin{theorem}\label{thm:cut-ext}
The rules $\cut$, $\bsub$ and $\wsub$ are admissible in $\CIMstar$.
\end{theorem}
\begin{proof}
We extend the proofs of the claims \textbf{(A)}, \textbf{(B)}, \textbf{(C)}, \textbf{(D)} in the proof of Theorem~\ref{th:cut}.
\textbf{(A)} is as before.
\textbf{(B, D)} Suppose the last rule $\Rule$ applied in the derivation of a premiss of $\cut$ is among $\rulen, \rulep, \ruled, \rulet$.
Then the cut formula is not principal in $\Rule$.
If $\Rule \in \{\rulen, \rulet\}$, we first apply $\cut$ on the premiss of $\Rule$ and the other premiss of $\cut$, then we apply $\Rule$. 
If $\Rule \in \{\rulep, \ruled\}$, we obtain the conclusion of $\cut$ by considering a different application of $\Rule$.
\textbf{(C)}
Suppose the premiss ($\nest{\bl{\phi}}, \D$) of $\bsub$, resp. the premiss ($\Lambda, \wh{\phi}$) of $\bsub$, is obtained by $\Rule \in \{\rulen, \rulep, \ruled, \rulet\}$.
If the block $\nest{\bl{\phi}}$, resp. $\wh{\phi}$, is not principal in $\Rule$, then the analysys is essentially as before.
Otherwise, we consider the following representative cases.

(i)
Assume the premiss ($\nest{\bl{\phi}}, \D$) of $\bsub$ is derived via $\rulep$ from ($\bl{\phi}$). 
Then $\Lambda$ is not empty in the other premiss ($\Lambda, \wh{\phi}$) of $\bsub$, otherwise both ($\bl{\phi}$) and ($\wh{\phi}$) would be derivable, against the soundness of the calculi.
Let $\Lambda = \bl{\psi}$.
By cutting on ($\bl{\psi}, \wh{\phi}$) and ($\bl{\phi}$), we obtain ($\bl{\psi}$), then by $\rulep$ we get ($\nest{\bl{\psi}}, \D$).

(ii)
Assume the premiss ($\Lambda', \nest{\bl{\psi}}, \nest{\wh{\phi}}$) of $\wsub$ is derived via $\ruled$ from ($\bl{\psi}, \wh{\phi}$). 
Then by cutting on ($\bl{\psi}, \wh{\phi}$) and the other premiss of $\wsub$ ($\bl{\phi}, \D$), we obtain ($\bl{\psi}, \D$), then by $\ruled$ we get ($\Lambda', \nest{\bl{\psi}}, \nest{\D}$).

(iii)
Assume the premiss ($\Lambda, \nest{\wh{\phi}}$) of $\wsub$ is derived via $\rulet$ from  ($\Lambda, \enest, \wh{\phi}$).
Consider the other premiss of $\wsub$ ($\bl{\phi}, \D$).
If $\D = \wh{\psi}$, then by cutting on ($\Lambda, \enest, \wh{\phi}$) and ($\bl{\phi}, \wh{\psi}$) we obtain ($\Lambda, \enest, \wh{\psi}$),
then by $\rulet$ we get ($\Lambda, \nest{\wh{\psi}}$).
If $\D = \bl{\psi}$, then by cutting on ($\Lambda, \enest, \wh{\phi}$) and ($\bl{\phi}, \bl{\psi}$) we obtain ($\Lambda, \enest, \bl{\psi}$),
then by $\bwk$, ($\Lambda, \enest, \nest{\bl{\psi}}, \bl{\psi}$), by $\nelim$, ($\Lambda, \nest{\bl{\psi}}, \bl{\psi}$),
finally by $\rulet$, ($\Lambda, \nest{\bl{\psi}}$).
Finally, if $\D = \emptyset$,  by cutting on ($\Lambda, \enest, \wh{\phi}$) and ($\bl{\phi}$) we obtain ($\Lambda, \enest$).
\end{proof}

\begin{theorem}[Completeness]
If $\vd_{\IMstar} \phi$, then $\vd_{\CIMstar} \wh{\phi}$.
\end{theorem}

\subsection{Calculi for extensions with $\axKbd$}\label{subsec:CIMK}

So far, blocks contained at most one formula. 
While not strictly necessary, it allowed us to simplify the calculus.
To formulate calculi for logics with $\axKbd$,
we need to adopt a more general formulation of the calculi
where blocks contain finite \emph{multisets} of marked formulas.
We sketch the required modifications.

We now call \emph{input blocks} those containing only input formulas
(or no formula at all),
and \emph{output blocks} those containing an output formula
(possibly together with some input formulas).
As before, sequents are single-succedent, meaning that their only output component is either an output formula,
or an output block, or is not present.
We extend the formula interpretation from Definition~\ref{def:fint} via:
  \begin{equation*}
    \fintf(\nest{ \ }) = \Box \top, \quad
    \fintf(\nest{\phi^\bullet_1, \ldots, \phi^\bullet_k}) = \Box (\phi_1 \land \dots \land \phi_k), \quad
    \fintf(\nest{\phi^\bullet_1, \ldots, \phi^\bullet_k, \wh{\psi}}) = \diam (\phi_1 \land \dots \land \phi_k \imp \psi)
  \end{equation*}
  For instance, $\fint(\nest{\bl{\phi}, \bl{\psi}}, \nest{\bl{\chi}, \wh{\theta}}) = \Box(\phi \land \psi) \imp \diam(\chi \imp \theta)$.
  
  Furthermore, in this setting the output pruning $\pr\G$ deletes output formulas
  that do not occur inside a block, and empties the \emph{entire} block if an
  output formula occurs inside it.
  So for instance we have  
  $\pr{(\bl{\phi}, \nest{\bl{\psi}}, \bl{\chi}, \wh{\theta})} = \bl{\phi}, \nest{\bl{\psi}}, \bl{\chi}$
  and
  $\pr{(\bl{\phi}, \nest{\bl{\psi}}, \nest{\bl{\chi}, \wh{\theta}})} = \bl{\phi}, \nest{\bl{\psi}}, \enest$.

\begin{definition}
The calculus $\CIMpK$ for $\IMpK$
is defined by the rules in Figure~\ref{fig:CIM},
where the restriction $|\D| \leq 1$ for the rules $\wbox$ and $\bdiam$ is dropped,
together with the following rule:
\begin{equation*}
\ax{$\G, \nest{\D, \Sigma}$}
\llab{$\rulek$}
\uinf{$\G, \nest{\D}, \nest{\Sigma}$}
\disp
\quad
\end{equation*}
\end{definition}

\begin{example}
  The axioms $\axKbox$ and $\axKdiam$ can be derived in $\CIMpK$ as follows:
  \begin{equation*}
\begin{small}
\ax{$\bl{\phi \to \psi}, \bl{\phi}, \wh{\phi}$}
\ax{$\bl{\psi}, \bl{\phi}, \wh{\psi}$}
\rlab{$\bimp$}
\binf{$\bl{\phi \to \psi}, \bl{\phi}, \wh{\psi}$}
\rlab{$\wbox$}
\uinf{$\nest{\bl{\phi \to \psi}, \bl{\phi}}, \wh{\Box\psi}$}
\rlab{$\rulek$}
\uinf{$\nest{\bl{\phi \to \psi}}, \nest{\bl{\phi}}, \wh{\Box\psi}$}
\rlab{$\bbox$}
\uinf{$\nest{\bl{\phi \to \psi}}, \bl{\Box\phi}, \wh{\Box\psi}$}
\rlab{$\bbox$}
\uinf{$\bl{\Box(\phi \to \psi)}, \bl{\Box\phi}, \wh{\Box\psi}$}
\rlab{$\wimp$}
\uinf{$\bl{\Box(\phi \to \psi)}, \wh{\Box\phi \to \Box\psi}$}
\rlab{$\wimp$}
\uinf{$\wh{\Box(\phi \to \psi) \to (\Box\phi \to \Box\psi)}$}
\disp
\qquad\quad
\ax{$\bl{\phi}, \bl{\phi \to \psi}, \wh{\phi}$}
\ax{$\bl{\phi}, \bl{\psi}, \wh{\psi}$}
\rlab{$\bimp$}
\binf{$\bl{\phi}, \bl{\phi \to \psi}, \wh{\psi}$}
\rlab{$\bdiam$}
\uinf{$\bl{\Diamond\phi}, \nest{\bl{\phi \to \psi}, \wh{\psi}}$}
\rlab{$\rulek$}
\uinf{$\nest{\bl{\phi \to \psi}}, \bl{\Diamond\phi}, \nest{\wh{\psi}}$}
\rlab{$\wdiam$}
\uinf{$\nest{\bl{\phi \to \psi}}, \bl{\Diamond\phi}, \wh{\Diamond\psi}$}
\rlab{$\bbox$}
\uinf{$\bl{\Box(\phi \to \psi)}, \bl{\Diamond\phi}, \wh{\Diamond\psi}$}
\rlab{$\wimp$}
\uinf{$\bl{\Box(\phi \to \psi)}, \wh{\Diamond\phi \to \Diamond\psi}$}
\rlab{$\wimp$}
\uinf{$\wh{\Box(\phi \to \psi) \to (\Diamond\phi \to \Diamond\psi)}$}
\disp
\end{small}
  \end{equation*}
\end{example}

\begin{proposition}\label{prop:sound k}
The rule $\rulek$ is sound with respect to $\IMpK$.
\end{proposition}
\begin{proof}
If both $\nest{\D}$ and $\nest{\Sigma}$ are input blocks, then soundness is immediate using $\axCbox$ (Lemma \ref{lemma:eq axioms}).
Suppose $\Sigma = (\Theta, \wh{\phi})$. Then $\G = \Lambda$. 
Let us shorten $(\bigwedge_{x \in \Lambda} \fintf(x)) = \lambda$;
$(\bigwedge_{y \in \Delta} \fintf(y)) = \delta$; and
$ (\bigwedge_{z \in \Theta} \fintf(z)) = \theta$.
We assume $\vd_{\IMpK} \fint(\Lambda, \nest{\D, \Sigma}) = 
\lambda \imp \diam (\delta \land \theta \imp \phi)$.
By the formula $\axW$ (Lemma \ref{lemma:eq axioms}), we have 
$\vd_{\IMpK} \Box\delta \land  \diam (\delta \land \theta \imp \phi)
\imp \diam (\delta \land (\delta \land  \theta \imp \phi))$.
Moreover, by $\mondiam$, 
$\vd_{\IMpK} \diam (\delta \land (\delta \land  \theta \imp \phi))
\imp \diam(\theta \imp \phi)$.
Combining everything, we get 
$\vd_{\IMpK} \lambda \land \Box\delta \imp \diam(\theta \imp \phi) =  \fint(\Lambda, \nest{\D}, \nest{\Sigma})$.
\end{proof}

The generalisation of the formalism to multisets within blocks slightly complicates the structural analysis of the calculus
(which motivates keeping it simple when possible).
Apart from adjusting all cases from the previous proofs with blocks possibly containing more than one formula
(which is in most cases simply a routine task that does not affect them significantly),
the main changes to the previous analysis are the following: 
(1) The rule $\nctr$ below is hp-admissible in $\CIMpK$;
in particular, the rules $\bctr$ and $\nctr$ are simultaneously proved admissible
by mutual induction on the height of the derivation of their premiss.
(2)
The rules $\bsub$ and $\wsub$ now take the form below.
\begin{equation*}
\ax{$\G, \nest{\D, \bl{\phi}, \bl{\phi}}$}
\llab{$\nctr$}
\uinf{$\G, \nest{\D, \bl{\phi}}$}
\disp
\qquad 
\ax{$\Lambda, \wh{\phi}$}
\ax{$\nest{\bl{\phi}, \Sigma}, \D$}
\llab{$\bsub$}
\binf{$\nest{\Lambda, \Sigma}, \D$}
\disp
\qquad
\ax{$\Lambda, \nest{\Sigma, \wh{\phi}}$}
\ax{$\bl{\phi}, \D$}
\llab{$\wsub$}
\binf{$\Lambda, \nest{\Sigma, \D}$}
\disp
\end{equation*}

\begin{proposition}\label{prop:k-hp-adm}
The rule $\rulek$ is height-preserving invertible in $\CIMpK$.
\end{proposition}
\begin{proof}
By induction on the height $h$ of the derivation of the conclusion $C = (\G, \nest{\D}, \nest{\Sigma})$ of $\rulek$.
As an example, suppose that $h > 0$ and consider the last rule $\Rule$ applied in the derivation of $C$.
If $\R = \rulek$ with $\nest{\D}, \nest{\Sigma}$ principal, then the claim trivially follows.
If $\R = \rulek$ with only $\nest{\D}$ principal, then $C = (\G', \nest{\D}, \nest{\Sigma}, \nest{\Theta})$ is
obtained from ($\G', \nest{\D, \Theta}, \nest{\Sigma}$).
Applying the induction hypothesis on this sequent we get ($\G', \nest{\D, \Theta,\Sigma}$), then by $\rulek$, ($\G', \nest{\D, \Sigma}, \nest{\Theta}$).
If $\R \neq \rulek$, we apply the induction hypothesis on the premiss of $\R$, then we apply $\R$.
\end{proof}

\begin{proposition}\label{prop:ext-K-hp-adm}
The rules $\bwk$, $\wwk$, $\bctr$, $\nctr$, $\nelim$ are height-preserving admissible in $\CIMpK$.
\end{proposition}
\begin{proof}
Let us just consider a few most significant cases in the proof of hp-admissibility of $\bctr$ and $\nctr$.
Suppose the premiss of $\bctr$, resp. $\nctr$, is derivable with height $h$, and that
$\bctr$ and $\nctr$ are hp-admissible up to height $h-1$.
If the premiss ($\G, \nest{\D}, \nest{\D}$) of $\bctr$ is obtained by $\rulek$ from ($\G, \nest{\D,\D}$),
then by $|\D|$ hp-applications of $\nctr$ (induction hypothesis) we obtain ($\G, \nest{\D}$).
If the premiss ($\G, \nest{\D}, \nest{\D}, \nest{\Sigma}$) of $\bctr$ is obtained by $\rulek$ from ($\G, \nest{\D, \Sigma}, \nest{\D}$),
then by hp-invertibility of $\rulek$ we have ($\G, \nest{\D, \Sigma, \D}$),
and by $|\D|$ hp-applications of $\nctr$ (induction hypothesis), ($\G, \nest{\D, \Sigma}$),
finally by $\rulek$, ($\G, \nest{\D}, \nest{\Sigma}$).
If the premiss ($\Lambda, \nest{\D, \bl{\phi}, \bl{\phi}}, \wh{\Box\psi}$) of $\nctr$ is obtained by $\wbox$ from  
($\D, \bl{\phi}, \bl{\phi}, \wh{\psi}$),
then by a hp-application of $\bctr$ (induction hypothesis) we obtain 
($\D, \bl{\phi}, \wh{\psi}$), and by $\wbox$ we get
($\Lambda, \nest{\D, \bl{\phi}}, \wh{\Box\psi}$).
\end{proof}

\begin{theorem}\label{th:cut k}
The rules $\cut$, $\bsub$, $\wsub$ are admissible in $\CIMpK$.
\end{theorem}
\begin{proof}
Despite the different formulation of $\bsub$ and $\wsub$,
the proof resembles that of Theorem~\ref{th:cut}.
The only significant changes concern the proof of claim 
\textbf{(C)}:
Assume $\forall h. Cut(c, h)$.
As before, admissibility of $\bsub$ is proved by induction on the height $j$ of the derivation of ($\nest{\bl{\phi}, \Sigma}, \D$).
Assume $j > 0$ and $\nest{\bl{\phi}, \Sigma}$ principal in the last rule $\Rule$ applied in the derivation of ($\nest{\bl{\phi}, \Sigma}, \D$).
If $\Rule = \wbox$, then we have ($\nest{\bl{\phi}, \Sigma}, \D', \wh{\Box \psi}$) obtained from ($\bl{\phi}, \Sigma, \wh{\psi}$).
By cutting $\bl{\phi}, \wh{\phi}$ from this sequent and the other premiss ($\Lambda, \wh{\phi}$) of $\bsub$,
we obtain ($\Lambda, \Sigma, \wh{\psi}$), then by $\wbox$, ($\nest{\Lambda, \Sigma}, \D', \wh{\Box \psi}$).
If $\Rule = \bdiam$, the case is similar.
If $\Rule = \rulek$, then we have ($\nest{\bl{\phi}, \Sigma}, \D', \nest{\Theta}$)
obtained from ($\nest{\bl{\phi}, \Sigma, \Theta}, \D'$).
Then by applying $\bsub$ to this sequent and the other premiss ($\Lambda, \wh{\phi}$) of $\bsub$
(which is admissible by the induction hypothesis), we get ($\nest{\Lambda, \Sigma, \Theta}, \D'$),
then by $\rulek$,  ($\nest{\Lambda, \Sigma}, \D', \nest{\Theta}$).
The admissibility of $\wsub$ is proved similarly.
\end{proof}

Finally, as a consequence of Proposition~\ref{prop:sound k}, Theorem~\ref{th:cut k} and derivability of $\axKbox$, $\axKdiam$ in 
$\CIMpK$:

\begin{theorem}[Soundness and completeness]
$\vd_{\IMpK} \phi$ if and only if $\vd_{\CIMpK} \wh{\phi}$.
\end{theorem}

%%%%%%%%%%%%%%%%%%%%%%%%%%%%%%%%%%%%%%%%%%%%%%%%%%%%%%%%%%%%%%%%%%%%%%%%%%%%%%%%
\section{Conclusion}

  We studied the intuitionistic monotone modal logic $\IM$ and several
  extensions, providing a semantics by means of constructive neighbourhood
  frames as well as cut-free proof systems.
  The proof-theoretical analysis reveals the following analogy between
  constructive and intuitionistic versions of $\EM$ and $\K$:
  \begin{quote}
    \it
    Constructive modal logics are single-succedent restrictions
    of Genzten calculi; \\
    Intuitionistic modal logics are single-succedent restrictions
    of structured sequent calculi.
  \end{quote}
  This immediately raises the question whether this pattern restricts to
  constructive and intuitionistic versions of $\EM$ and $\K$, or if
  it applies more generally.

  Furthermore, Arisaka, Das and Strassburger~\cite{AriDasStr15} provide a nested
  calculus for $\CK$ which requires a more radical form of pruning than that
  employed for $\IK$.
  Similarly, we turned the calculus for $\IM$ into one for $\WM$ by replacing
  context pruning with context thinning.
  Again, it would be interesting to investigate if this is merely a coincidence
  or part of a pattern.

  Other interesting questions include:
  Does the semantic reading from Remark~\ref{rem:sem} allow one to extract
  countermodels for failed derivations in $\CIM$?
  What other variants of context thinning and pruning can be defined,
  and what logics do they give rise to?
  What kind of intuitionistic variants of $\EM$ arise from the
  single-succedent restriction of other structured calculi, such as
  the nested calculus from~\cite{LelPim19}?
  We hope to address some of these in future research.
  
  Finally, we would like to study the complexity of the derivability problem in $\IM$ and its extensions by either extending an optimal calculus for $\IPL$ like the one from~\cite{Hud93}, or considering reductions to classical multi-modal logics
  similarly to~\cite{Dal25,Fis77},~\cite[Chapter~10.2]{YB}.

\clearpage
%% Bibliography
%% Make sure to use the bibliographystyle aiml22.
\bibliographystyle{eptcs}
\bibliography{aiml}

\end{document}